%% file: cameraReady.tex
\documentclass[runningheads]{llncs}
\usepackage[utf8]{inputenc}
\DeclareMathAlphabet{\mathbfit}{OML}{cmm}{b}{it}
\usepackage{graphics,graphicx,color}
  
\usepackage{amsthm}
\usepackage{microtype,verbatim,dsfont}
\usepackage{xspace,enumerate,paralist}
\usepackage{amsmath,amssymb}
\usepackage{subcaption}
\usepackage{paralist}
\usepackage[linesnumbered,noend,ruled]{algorithm2e}
\usepackage{todonotes}
\usepackage{hyperref}

\newcommand{\myparagraph}[1]{\smallskip\noindent\textbf{#1}}
\newcommand{\pent}{\ensuremath \mathcal{P}}
\newcommand{\emb}{\ensuremath \mathcal{R}}
\newcommand{\textemb}{optimal $2$-planar rotation scheme\xspace}
\newcommand{\textembs}{optimal $2$-planar rotation schemes\xspace}
\newcommand{\partialemb}[1][]{
\ifthenelse{\equal{#1}{}}
{\ensuremath \mathcal{R}_{\mathcal{T}}}
{\ensuremath \mathcal{R}_{\mathcal{T}#1}}%
}

\newcommand{\triplet}{triplet\xspace}
\newcommand{\triplets}{{\triplet}s\xspace}
\newcommand{\facialKFive}{facial $5$-clique\xspace}

\newcommand{\facialKFives}{facial $5$-cliques\xspace}
\newcommand{\crossingConfiguration}{\textsf{CROSS-BAT}\xspace}
\newcommand{\bad}{bad\xspace}
\newcommand{\notbad}{not bad\xspace}
\newcommand{\compliant}{$\emb_p$-compliant\xspace}
\newcommand{\Oh}{\ensuremath \mathcal{O}}
\newtheorem{observation}{Observation}
\newcommand{\crossing}{clearly crossing\xspace}

\theoremstyle{definition}
\newtheorem{details}{Details}
\newcommand{\myhfil}{\hfil}

\numberwithin{lemma}{section}
\numberwithin{proposition}{lemma}
\numberwithin{property}{section}
\numberwithin{observation}{section}
\numberwithin{corollary}{section} 

\newcommand{\ver}{arxiv}
\newcommand{\arxapp}[2]{\ifthenelse{\equal{\ver}{conf}}{#2}{#1}}

 \SetKwComment{ncc}{// }{}
 \SetCommentSty{mytexttt}
 
 \usepackage{thm-restate}

\title{Recognizing and Embedding\texorpdfstring{\\}{}Simple Optimal \texorpdfstring{$2$}{2}-Planar Graphs\thanks{The work of C.\ Raftopoulou is funded by the National Technical University of Athens research program
$\mathrm{\Pi}$EBE 2020. H.\ F\"orster and M.\ Kaufmann are supported by DFG grant KA812-18/2.}}
\author{
Henry Förster\inst{1}\orcidID{0000-0002-1441-4189} \and Michael Kaufmann\inst{1}\orcidID{0000-0001-9186-3538} \and Chrysanthi N. Raftopoulou\inst{2}\orcidID{0000-0001-6457-516X}}
\institute{WSI, T\"ubingen, Germany\\
\email{\{%bekos,
foersth,mk\}@informatik.uni-tuebingen.de}\and
NTUA, Athens, Greece\\
\email{crisraft@mail.ntua.gr}}

\authorrunning{F\"orster, Kaufmann and Raftopoulou}
\begin{document}

\maketitle

%\linenumbers
\begin{abstract}
In the area of beyond-planar graphs, i.e.\ graphs that can be drawn with some local restrictions on the edge crossings, the recognition problem is prominent next to the density question for the different graph classes. For 1-planar graphs, the recognition problem has been settled, namely  it is NP-complete for the general case, while optimal 1-planar graphs, i.e.  those with maximum density, can be recognized in linear time. For 2-planar graphs, the picture is less complete.
As   expected,  the recognition problem has been found to be NP-complete in general. In this paper, we consider the recognition of simple optimal 2-planar graphs. We exploit a combinatorial characterization of such graphs and present a linear time algorithm for recognition and embedding.
\end{abstract}

\section{Introduction}\label{sec:intro}

In the field of graph drawing beyond planarity, researchers study graphs that admit drawings with crossings only in restrictive local configurations~\cite{DBLP:journals/tcs/BinucciGDMPST15,DBLP:journals/algorithmica/CheongHKK15,DBLP:journals/tcs/DidimoEL11,DBLP:journals/cj/GiacomoDLMW15,KU14}. 
Mostly studied are  \emph{1-planar graphs}, which admit drawings where each~edge~is crossed at most once. Classic results on these graphs concern mostly the density~\cite{BSW84,Ringel65}, recently, other problems such as generation~\cite{DBLP:journals/siamdm/Suzuki10}, characterization~\cite{DBLP:conf/cocoon/HongELP12}, recognition~\cite{AuerBBGHNR16,DBLP:journals/corr/Brandenburg16a,Brandenburg18a,Brandenburg18,BrandenburgDEKL16}, coloring~\cite{Borodin95} and page number~\cite{BB0R15,DBLP:journals/jacm/DujmovicJMMUW20} have been studied.
The most natural extension of $1$-planar graphs, is the family of \emph{$k$-planar graphs}, that is graphs that admit drawings where each edge is crossed at most $k$ times. While there are plenty of results on $1$-planarity,  $k$-planarity turns out to be more challenging. 
The first outstanding results on the density of simple $k$-planar graphs came with the work of Pach and T\'{o}th~\cite{PachT97} for $k\leq 4$. 
In particular, they proved that simple $2$-planar graphs have at most $5n -10$ edges. 
Later on, improved upper bounds of  $5.5n - 11$~\cite{PachRTT06} and $6n-12$~\cite{DBLP:journals/corr/Ackerman15} have been shown for simple  $3$- and $4$-planar graphs, respectively. 
The bounds for $2$- and $3$-planar graphs are tight even for non-simple graph without homotopic parallel edges and self-loops~\cite{DBLP:conf/compgeom/Bekos0R17}. 

While  testing for planarity can be done in linear time, and several classes of beyond-planar graphs have a linear number of edges,
the corresponding recognition problems are NP-hard, which impedes the practical application of related results.  
In particular, it was recently proven that the recognition of $k$-planar graphs is NP-hard for every value of $k\geq 1$~ \cite{urschel20}. 
Still, optimal 1-planar graphs, which form a prominent subclass of $1$-planar graphs, can be recognized in linear time~\cite{Brandenburg18}, while recognition of  maximal $1$-planar graphs and $4$-maps takes cubic
time~\cite{DBLP:journals/tcs/EadesHKLSS13,DBLP:journals/algorithmica/Brandenburg19}. 
Previously, efficient recognition algorithms have been developed for more restricted subclasses of $1$-planar graphs, see e.g.~\cite{AuerBBGHNR16,Brandenburg18a,BrandenburgDEKL16,DBLP:journals/algorithmica/HongEKLSS15}.

For $k\geq2$, there are significantly fewer results. 
A notable result is the linear time recognition algorithm for full outer-2-planar graphs~\cite{HongN19}, while   for $k\geq3$ the literature still lacks similar results. 

\myparagraph{Our contribution.} In this paper, we extend the study on the recognition of subclasses of $k$-planar graphs. Namely, we prove that simple optimal $2$-planar graphs, that is simple $2$-planar graphs with $5n-10$ edges, can be efficiently recognized in $\Oh(n)$ time. In our strategy we first remove a set of edges that must be involved in crossings in any optimal $2$-planar drawing of the given graph. %
We unwrap specific properties of the remaining graph that allow us to find a unique drawing and  insert the deleted crossing edges. 

\myparagraph{Paper organization.} In Section~\ref{sec:preliminaries} we give~preliminary notations  and properties of optimal $2$-planar graphs. Section~\ref{sec:crossBat} presents \crossingConfiguration, an important $2$-planar substructure,  and examines its structural properties. Section~\ref{sec:triplets} focuses on the 
insertion of the deleted crossing edges.
Our recognition algorithm in Section~\ref{sec:algo} 
combines all previous  results. We conclude in Section~\ref{sec:conclusions} with further research directions.
Due to space constraints, several proofs and details, indicated by $(\ast)$, \arxapp{are deferred to an appendix}{are omitted and can be found in the full version~\cite{arxivVersion}}.

\section{Preliminaries}\label{sec:preliminaries}
The degree of a vertex $v \in V$ is denoted by $d(v)$ and its set of neighbors by $N(v)$. For  $u,v\in V$ we  write $N(u,v)$ to denote the set of common neighbors of $u$ and $v$, that is $N(u,v) =N(u) \cap N(v)$. 
A \emph{drawing} of a graph maps the vertices to points in the plane and the edges to simple open Jordan curves connecting the points corresponding to their endpoints.
A drawing is \emph{$k$-planar} if each curve representing an edge is crossed at most $k$ times by the other  edges together.  
A graph $G$ is $k$-planar if it admits a $k$-planar drawing.

The topology of a planar drawing  is described by a  \emph{rotation scheme}. For $k$-planarity, a \emph{$k$-planar rotation scheme} specifies the counter-clockwise circular order of edges around each vertex and, additionaly,  a sequence of at most $k$ crossings along each edge. Hence, a $k$-planar rotation scheme defines an equivalence class of $k$-planar drawings. 

For a graph $G$, let $H$ a subgraph of $G$ and let $\emb_H$ be a ($k$-planar) rotation scheme of $H$. We say that $\emb_H$ is a \emph{partial ($k$-planar) rotation scheme} of $G$. We further say that $\emb_H$ \emph{can be extended} to a ($k$-planar) rotation scheme $\emb$ of $G$ if the restriction of $\emb$ on $H$ is $\emb_H$. Conversely, we say that $\emb$ \emph{extends} $\emb_H$. 
Let $C$ be a cycle of  graph $G$,  whose edges define a simple closed curve in a drawing $\Gamma$ of $G$. We refer to the region to the left of this curve in a counter-clockwise walk along $C$ in $\Gamma$ as the \emph{interior of the region} bounded~by~cycle~$C$. Note that given drawing $\Gamma$, the interior of a region may be unbounded.

A $2$-planar graph that reaches the upper bound of $5n-10$ edges is called \emph{optimal $2$-planar}. Note that the value of $5n-10$ edges can be reached only when $n\equiv 2\mod 3$. Each induced subgraph $H$ of a $2$-planar graph $G$ is also 2-planar and thus has at most $5|V(H)|-10$ edges.  This implies that if $G$ is 2-planar then it is also \emph{9-degenerate}, i.e. there is a vertex ordering $v_1,\ldots,v_n$  such that, for each $i=1,\ldots,n$, vertex $v_i$ of the induced subgraph $G_i=G[v_1,\ldots,v_i]$ has degree at most 9 in $G_i$.
We call such a vertex ordering  a \emph{9-degenerate} sequence of $G$.
 
In this work, we consider only \emph{simple} graphs which we assume from now on without explicitly stating it. The structure of simple optimal $2$-planar graphs is well established.  Namely, an \textemb $\emb$ of a simple optimal $2$-planar graph specifies a crossing-free $3$-connected\footnote{The 3-connectivity derives from the fact that the graph is simple.} spanning pentangulation $\pent$ and a set of crossing edges~\cite{DBLP:journals/dm/BekosGDLMR19,DBLP:conf/compgeom/Bekos0R17}. The edges around each vertex are cyclically ordered so that every crossing-free edge (belonging to $\pent$) is followed by  two crossing edges, which are succeeded by one crossing-free edge, and so on. 
Hence the degree of every vertex is a multiple of three. Moreover, the   minimum degree of a simple optimal $2$-planar graph is $9$ since  $\pent$ is $3$-connected, i.e., every vertex is incident to at least $3$ crossing-free edges. 
Now, each pentangular face $f$ of $\pent$ contains exactly five crossing edges of $\emb$ which form a $5$-clique in $G$ with the boundary edges of $f$; 
we call such a $5$-clique a \emph{facial $5$-clique}.  Note that not all $5$-cliques of $G$ are facial in $\emb$. This is a major difficulty for recognizing optimal $2$-planar graphs. We  denote a facial $5$-clique of $\emb$ with outer cycle $(c_0,\ldots,c_4)$  as $\langle c_0,\ldots,c_4 \rangle$ and say that the crossing edges $(c_i,c_j)$ with $(i-j)\mod 5 \in \{2,3\}$ \emph{belong to} $\langle c_0,\ldots,c_4 \rangle$. As $\pent$ is $3$-connected, it suffices to detect all \facialKFives to find the optimal $2$-planar rotation scheme of a given input graph, if one exists. 

In the following, we state some preliminary observations and properties that we will use in the remainder of the paper. In particular, Properties~\ref{lem:crossingEdgesBelongToSameK5}--\ref{lem:crossingEdgesAreNotShared} immediately arise from the combinatorial characterization of simple optimal 2-planar graphs in \cite{DBLP:conf/compgeom/Bekos0R17} also summarized above. 
So, let $G$ be an optimal $2$-planar graph with an \textemb $\emb$.

\begin{property}\label{lem:crossingEdgesBelongToSameK5}
% Let $G$ be an optimal $2$-planar graph with an \textemb $\emb$. 
A pair of crossing edges in $\emb$ belongs to the same \facialKFive.
\end{property}

\begin{property}\label{lem:crossingTwoEdgesBelongToSameK5}
% Let $G$ be an optimal $2$-planar graph with an \textemb $\emb$. 
If in $\emb$, edge  $e$ is crossed by two edges $e_1$ and $e_2$, then $e_1$ and $e_2$ share a common endpoint. Also, $e$, $e_1$ and $e_2$  belong to the same \facialKFive.
\end{property}

\begin{property}\label{obs:twoK5sShareOnly2Vertices}
% Let $G$ be an optimal $2$-planar graph with an \textemb $\emb$. 
Any two \facialKFives in $\emb$ share at most two vertices.
\end{property}

\begin{property}\label{lem:crossingEdgesAreNotShared}
% Let $G$ be an optimal $2$-planar graph, and let $\emb$ be an \textemb of $G$. 
A crossing edge of $\emb$ belongs to exactly one \facialKFive and a crossing-free edge of $\emb$ belongs to exactly two \facialKFives.
\end{property}

Combining the above properties we can further prove the following:

\begin{restatable}[$\mathbf{\ast}$]{corollary}{crossedEdgesShareAKFive}
\label{lem:crossedEdgesShareAK5}
For two edges $(u,v)$ and $(u,w)$  of an optimal $2$-planar graph $G$ let $S=N(u,v)\cup N(u,w)\cup\{u,v,w\}$. If $|S| < 8$, then $(u,v)$ and 
$(u,w)$ belong to the same \facialKFive in any \textemb of $G$.
\end{restatable}

\begin{corollary}\label{col:edgesIncidentToCrossedEdgeDoNotCross}
Let $G$ be an optimal $2$-planar graph with an \textemb $\emb$. Let $(u,v)$ be a crossing edge inside a facial $5$-clique $C$ in $\emb$. Further, let $u' \not\in C$ be a neighbor of $u$ and $v' \notin C$ be a neighbor of $v$. Then, edges $(u,u')$ and $(v,v')$ do not cross in $\emb$; see Fig.~\ref{fig:crossingEdgeInAK5:1}. 
\end{corollary}

\begin{figure}[t]
    \centering
    \includegraphics[page=16]{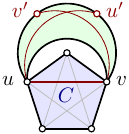}
    \caption{Forbidden configuration excluded by Corollary~\ref{col:edgesIncidentToCrossedEdgeDoNotCross}.}
    \label{fig:crossingEdgeInAK5:1}
\end{figure}

Consider an \textemb $\emb$ of $G$. If edge $(u,v)$ is crossing-free in $\emb$, then by Property~\ref{lem:crossingEdgesAreNotShared} $u$ and $v$ belong to two \facialKFives.
Hence, $|N(u,v)| \geq 6$ holds. Although this condition is not sufficient to conclude that an edge  is crossing-free,  we can use it to  identify some crossing edges:

\begin{corollary}\label{col:identificationProperty}
Let $G$ be an optimal 2-planar graph and $(u,v) \in E(G)$. Edge $(u,v)$ is crossing in any \textemb of $G$ if  $|N(u,v)| < 6$.
\end{corollary}

Our recognition algorithm relies on the identification of crossing edges based on 
Corollary~\ref{col:identificationProperty}. We say that an edge $(u,v)$ is \emph{\crossing} if and only if $|N(u,v)| < 6$. Otherwise  $(u,v)$ is \emph{potentially planar}. Note that in an \textemb, a potentially planar edge is not necessarily crossing-free. However, every crossing-free edge of any \textemb is potentially planar.
We define as $G_p$ the subgraph of $G$ formed by all potentially planar edges. Graph $G_p$ is $3$-connected and spans $G$  as  the corresponding crossing-free pentangulation of each \textemb 
is a subgraph of~$G_p$.

\section{The \crossingConfiguration configuration}
\label{sec:crossBat}
\input{crossBat.tex}

\section{Identifying Facial $5$-Cliques}\label{sec:triplets}

Assume that we have fixed the rotation scheme for every identified instance of \crossingConfiguration. 
Let $G_p$ be the 
spanning subgraph of $G$ formed by all potentially planar edges after the 		reclassification process. As discussed in Section~\ref{sec:crossBat}, $G_p$ is planar and $3$-connected, i.e. it has a unique planar rotation scheme $\emb_p$. Furthermore, by Lemma~\ref{lem:crossingConfigurationChooseEmbed}, if $G$ is optimal $2$-planar, it admits an \textemb that extends $\emb_p$ which we call \emph{\compliant}. 
For any \compliant \textemb $\emb$, $G_p$ contains the corresponding spanning pentangulation $\pent$ as a subgraph. 
Hence each face of $\emb_p$  has length $3$, $4$ or $5$ and is part of a facial $5$-clique. Hence, we can arbitrarily triangulate them and assume from now on that $G_p$ is triangulated. 
Triangulating a face of $\emb_p$ corresponds to reclassifying some chords from \crossing to potentially planar.

Let $(f_1,f,f_2)$ be a path in the dual $G_p^*$ of $G_p$ so that $f_1 = (u,w_1,v_1)$, $f=(u,w_2,w_1)$ and $f_2=(u,v_2,w_2)$, as shown in Fig.~\ref{fig:triplet_base}. If the subgraph induced by the vertices $\{u,v_2,w_2,w_1,v_1\}$ is a $5$-clique in $G$, we call $T=\langle f_1,f,f_2 \rangle$ a \emph{\triplet}. $T$ contains vertices $u,v_2,w_2,w_1,v_1$, faces $f_1,f,f_2$ and the three \crossing edges $e_1=(v_1,w_2)$, $e_2=(v_2,w_1)$ and $e=(v_1,v_2)$, as shown in Fig.~\ref{fig:triplet_base}. We say that $e_1$, $e_2$ and $e$ \emph{belong} to \triplet $T$.

In any \compliant \textemb $\emb$ (if it exists), faces and \crossing edges of $\emb_p$ are partitioned into \triplets, such that every face and \crossing edge belongs to exactly one of these \triplets. Furthermore, the \triplets are in  1-1 correspondence with the facial $5$-cliques of $\emb$. We say that $e$, $e_1$ and $e_2$ are \emph{assigned} to the \triplet  $T=\langle f_1,f,f_2\rangle$  \emph{w.r.t.} $\emb$ if $T$ induces a \facialKFive in $\emb$. Similarly, we say that faces $f_1$, $f$ and $f_2$ of $T$ are \emph{assigned} to $T$. If such an assignment is not possible, then $G$ is not optimal $2$-planar.
\begin{figure}[t]
    \centering
    \includegraphics[width=0.3\textwidth,page=1]{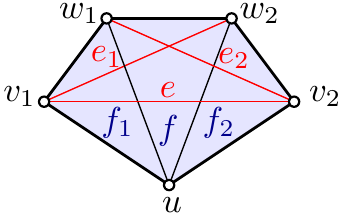}
    \caption{Illustration of a  \triplet $\langle f_1,f,f_2\rangle$, along with its \crossing edges $e_1$, $e_2$ and $e$. Edges of $G_p$ are drawn black, \crossing edges of $G$ are red.}
    \label{fig:triplet_base}
\end{figure}

Let $\mathcal{T}$ be a set of \triplets such that any two \triplets in $\mathcal{T}$ are face-disjoint  and contain different \crossing edges. Consider the partial $2$-planar rotation scheme $\emb_\mathcal{T}$ of $G$ that 
\begin{inparaenum}[(i)]
\item extends $\emb_p$, 
\item 
the \crossing edges of each $T \in \mathcal{T}$ are assigned to $T$, and
\item there is no other assignment of \crossing edges. 
\end{inparaenum}We say that $\mathcal{T}$ is \emph{\bad} if and only if $\emb_\mathcal{T}$ cannot be extended to a \compliant \textemb $\emb$ of $G$. 
Our goal is to find an assignment of all \crossing edges of $G$ to a set of \triplets $\mathcal{T}$ such that 
$\mathcal{T}$ is \notbad
if $G$ is optimal 2-planar. We actually prove a stronger result, namely that the set $\mathcal{T}$ is unique.
We will use two observations that follow from the simplicity of $G$:

\begin{observation}\label{obs:eitherOnlyTwoVerticesOrFace}
Let $f=(u,v,w)$ be a triangular face of $\emb_p$ and let $T$ be a \triplet that contains vertices $u$, $v$ and $w$. Then $T$ contains face $f$.
\end{observation}

\begin{observation}\label{obs:crossBothFacesOrNone}
Let $f=(u,v,w)$ and $f'=(u',w,v)$ be two adjacent faces of $\emb_p$ and $(u,u')$  a \crossing edge. Let $\emb$ be a  \compliant \textemb and $T$ be the \triplet that $(u,u')$ is assigned to w.r.t.\ $\emb$. Then, either 
\begin{inparaenum}[(i)]
\item $(u,u')$ is drawn inside $f$ and $f'$ in $\emb$ and $T$ contains both $f$ and $f'$,  or
\item  $(u,u')$  is drawn outside $f$ and $f'$ in $\emb$ and $T$ contains neither $f$ nor $f'$.
\end{inparaenum}
\end{observation}

If any of
$e_1$, $e_2$ or $e$ belongs only to one \triplet $T$, then $T$ is a facial $5$-clique in any \compliant \textemb $\emb$ and $e_1$, $e_2$ and $e$ can be assigned to $T$. So, we assume 
that each of $e_1$, $e_2$ and $e$ belong to at least one triplet different from $T$. Let $T_1$ and $T_2$ be two  \triplets that contain edges $e_1$ and $e_2$ respectively and are different from $T$. Note that $T_1=T_2$ might hold.
Let $\mathcal{T}=\{T\}$. If for each such pair of \triplets $T_1$ and $T_2$ 
we can conclude that the set $\mathcal{T}'=\{T_1,T_2\}$ is \bad, then in any \compliant \textemb $\emb$ of $G$ (if it exists) edges $e_1$, $e_2$ and $e$ must be assigned to $T$. In the following, we compare  $\mathcal{T}$ against all possible sets $\mathcal{T}'$ and prove that at least one of $\mathcal{T}$ and $\mathcal{T}'$ is \bad. This allows to decide, for each \triplet $T$, if $T$ forms a \facialKFive in every \compliant \textemb of $G$ or in none. 
Note that when we write $\mathcal{T}=\{T\}$ and $\mathcal{T}'=\{T_1,T_2\}$, we assume~that $T= \langle f_1,f,f_2 \rangle$ as shown in Fig.~\ref{fig:triplet_base} and \triplets $T_1$ and $T_2$ contain edges $e_1$ and $e_2$, respectively.

We first restrict how \triplets $T_1$ and $T_2$ relate to \triplet $T$.
Observation~\ref{obs:crossBothFacesOrNone} applied to $T_1$ implies that if $T_1$ shares either $f$ or $f_1$ with $T$, then, in the presence of edge $e_1$, $T_1$ shares both $f$ and $f_1$ with $T$. A symmetric argument applies for $T_2$. 
It follows that one of the following cases holds for $\{i,j\}=\{1,2\}$: 
\begin{enumerate}[\bf ~~C.1]
\item \label{c:1} $T_i$ is face-disjoint with $T$, or
\item \label{c:3} $T_i$ shares only face $f_j$ with $T$, or,
\item \label{c:2} $T_i$ shares only faces $f$ and $f_i$ with $T$.
\end{enumerate}

In the next two lemmas, we show that every $\mathcal{T}'=\{ T_1, T_2\}$ is \bad if Case C.\ref{c:1} applies  either for none or for both of $T_1$ and $T_2$.

\begin{restatable}[$\mathbf{\ast}$]{lemma}{oneTripletFaceDisjoint}\label{lem:oneTripletFaceDisjoint}
Let $\mathcal{T}' =\{ T_1, T_2\}$. If none of $T_1$ and $T_2$ is face-disjoint from $T$, then $\mathcal{T}'$ is \bad.
\end{restatable}

\begin{restatable}[$\mathbf{\ast}$]{lemma}{notBothTripletsFaceDisjoint}\label{lem:notBothTripletsFaceDisjoint}
Let $\mathcal{T}' =\{ T_1, T_2\}$. If both $T_1$ and $T_2$ are face-disjoint from $T$, then $\mathcal{T}'$ is \bad.
\end{restatable}

\begin{proof}[Sketch]
Assuming that $\mathcal{T}'$ is \notbad, we arrive at the configuration shown in Fig.~\ref{fig:shortTriplet_proof1}. Since $(w_1,w_2)$ is crossing, Corollary~\ref{col:edgesIncidentToCrossedEdgeDoNotCross} is violated.
\end{proof}

By Lemmas~\ref{lem:oneTripletFaceDisjoint} and~\ref{lem:notBothTripletsFaceDisjoint}, 
Case~C.\ref{c:1} applies for exactly one of the two triplets $T_1$ and $T_2$
and $T_1 \neq T_2$. Then, the other triplet complies with Case~C.\ref{c:3} or with Case~C.\ref{c:2}. For each possible combination we  first prove structural properties arising from the assumption that $\mathcal{T'}$ is \notbad, and then show that these new restrictions make $\mathcal{T}$ \bad. In the next two lemmas, we consider the setting, where one of $T_1$ and $T_2$ complies with Case~C.\ref{c:1} while the other one complies with~C.\ref{c:3}.  

\begin{figure}[t]
    \centering
    \begin{subfigure}[b]{0.4\textwidth}
    \includegraphics[page=2]{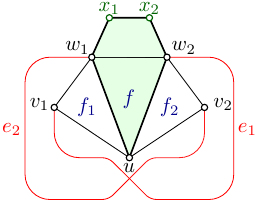}
    \caption{}
    \label{fig:shortTriplet_proof1}
    \end{subfigure}
    \myhfil
    \begin{subfigure}[b]{0.35\textwidth}
    \includegraphics[page=3]{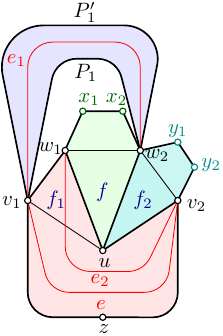}
    \caption{}
    \label{fig:shortTriplet_proof2}
    \end{subfigure}
    
    \caption{Illustrations for the proofs of (a)~Lemma~\ref{lem:notBothTripletsFaceDisjoint}, and,  (b)~Lemma~\ref{lem:c1c3_conditions}.}
    \label{fig:triplet_2}
\end{figure}

\begin{restatable}[$\mathbf{\ast}$]{lemma}{cOneCThreeConditions}\label{lem:c1c3_conditions}
Let $\mathcal{T}'=\{T_1,T_2\}$, such that $T_1$ is face-disjoint from $T$ and $T_2$ shares only face $f_1$ with $T$. 
For $u \in V(T)$, if \begin{inparaenum}[(i)]
\item $d(u)>9$, or,
\item for every vertex $x\in S=N(u)\setminus V(T)$ it holds $|N(x)\cap S|\geq 2$,
\end{inparaenum} then $\mathcal{T}'$ is \bad.
\end{restatable}

\begin{proof}[Sketch]
Assuming that $\mathcal{T}'$ is \notbad and $d(u)=9$, we arrive at the configuration shown in Fig.~\ref{fig:shortTriplet_proof2}. We can conclude that not all conditions of the lemma hold.
\end{proof}

\begin{restatable}[$\mathbf{\ast}$]{lemma}{cOneCThree}\label{lem:c1c3}
Let $\mathcal{T}=\{T\}$ and let $v \in V(T)$ such that $d(v) = 9$. If there exists a vertex $x\in S=N(v)\setminus V(T)$ such that $|N(x)\cap S|\leq 1$, then $\mathcal{T}$ is \bad.
\end{restatable}

By Lemmas~\ref{lem:c1c3_conditions} and~\ref{lem:c1c3} it  remains to consider the case where one of $T_1$ and $T_2$ complies with Case~C.\ref{c:1} while the other complies with Case~C.\ref{c:2}.
So, we assume w.l.o.g. that $T_1$ is face-disjoint with $T$, while $T_2$ shares faces $f$ and $f_2$ with $T$. Recall that \crossing edge $e$ belongs to $T$. If $\mathcal{T}'=\{T_1,T_2\}$ is \notbad $e$ belongs to  a \triplet $T_e\neq T$ 
such that  $\mathcal{T}'\cup\{T_e\}$ is \notbad. 
Note that $T_e\neq T_2$ as otherwise vertex $v_1$ belongs to $T_2$ and $T_2=T$.
Hence,  either $T_e\neq T_1$ or $T_e=T_1$. 

The following four lemmas investigate the subcase $T_e\neq T_1$. Note that $T_1$ and $T_e$ are face-disjoint as otherwise  $\mathcal{T}'=\{T_1,T_2,T_e\}$ is \bad. 
The first two lemmas examine the scenario where $T_1$  contains vertex $w_1$ or $T_e$ contains vertex $u$ of $T$. 

%%%%%%%%%%%%%%%%%%%%%%%%%%%%%%%%%%%%%%%%%%%%%%%
\begin{restatable}[$\mathbf{\ast}$]{lemma}{containsVertexConditions}\label{lem:c1c2TeContainsVertex_conditions}
Let $\mathcal{T}'=\{T_1,T_2,T_e\}$, such that $T_1$ is face-disjoint from $T$, $T_2$ shares faces $f$ and $f_2$ with $T$, and $T_e$ is face-disjoint from $T_1$. Then:
\begin{itemize}
    \item[--] If 
    $T_1$ contains vertex $w_1$ of $T$ and additionally
    \begin{inparaenum}[(i)]
    \item $d(w_1)>9$, or, 
    \item there is a vertex $x\in S=N(w_1)\setminus V(T)$ such that $|N(x)\cap S|\geq 4$
    \end{inparaenum}, then $\mathcal{T}'$ is \bad. 
    \item[--] If 
    $T_e$ contains vertex $u$ of $T$ and additionally
    \begin{inparaenum}[(i)]
    \item $d(u)>9$, or, 
    \item there is a vertex $x\in S=N(u)\setminus V(T)$ such that $|N(x)\cap S|\geq 4$
    \end{inparaenum}, then $\mathcal{T}'$ is \bad. 
\end{itemize}
\end{restatable}

\begin{proof}[Sketch]
Assuming that $\mathcal{T}'$ is \notbad and $d(w_1)=9$, in the first case, we arrive at the configuration shown in Fig.~\ref{fig:shortProofs:c1c2ContainsVertex}. We  conclude that not all conditions hold.
\end{proof}

\begin{figure}[t]
\centering
    \begin{subfigure}[b]{0.28\textwidth}
    \includegraphics[page=1]{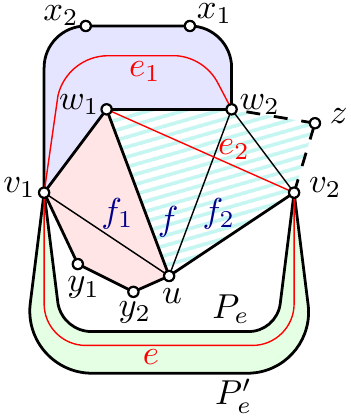}
    \caption{}
    \label{fig:shortProofs:c1c2ContainsVertex}
    \end{subfigure}
    \myhfil
    \begin{subfigure}[b]{0.32\textwidth}
    \includegraphics[page=8]{figures/optimal_c1c2.pdf}
    \caption{}
    \label{fig:shortProofs:c1c2NoVertex}
    \end{subfigure}
    \myhfil
    \begin{subfigure}[b]{0.35\textwidth}
    \includegraphics[page=3]{figures/doubleKfive.pdf}
    \caption{}
    \label{fig:shortProofs:doubleK5}
    \end{subfigure}
    \caption{Illustration for the proofs of Lemmas 
    (a)~\ref{lem:c1c2TeContainsVertex_conditions},
    (b)~\ref{lem:c1c2TeNoVertex}, and,
    (c)~\ref{lem:noDoubleK5}.}
    \label{fig:shortProofs}
\end{figure}

\begin{restatable}[$\mathbf{\ast}$]{lemma}{containsVertex}\label{lem:c1c2TeContainsVertex}
Let $\mathcal{T}=\{T\}$ and let $v \in V(T)$ with $d(v)=9$. If for every vertex $x\in S=N(v)\setminus V(T)$ it holds $|N(x)\cap S|\leq 3$, then $\mathcal{T}$ is \bad.
\end{restatable}

The next two lemmas consider the scenario, where $T_1$ does not contain vertex $w_1$ and $T_e$ does not contain vertex $u$ of \triplet $T$.

\begin{restatable}[$\mathbf{\ast}$]{lemma}{noVertexConditions}\label{lem:c1c2TeNoVertex_conditions}
Let $\mathcal{T}'=\{T_1,T_2,T_e\}$, such that $T_1$ is face-disjoint from $T$, $T_2$ shares faces $f$ and $f_2$ with $T$, and $T_e$ is face-disjoint from $T_1$. Assume that $T_1$ does not contain vertex $w_1$ and $T_e$ does not contain vertex $u$ of $T$. If there is no \triplet $T_{f_1}$ that contains $f_1$ such that \begin{inparaenum}[(i)]
\item $T_{f_1}$ shares only vertices of $f_1$ with $T$, and, \item $T_{f_1}$ is face-disjoint from all of $T_1$, $T_2$ and $T_e$\end{inparaenum}, then $\mathcal{T}'$ is \bad.
\end{restatable}

\begin{restatable}[$\mathbf{\ast}$]{lemma}{noVertex}\label{lem:c1c2TeNoVertex}
Let \triplets $T_1$, $T_2$, $T_e$ be pairwise disjoint such that $T_1$ is face-disjoint from $T$, $T_2$ shares faces $f$ and $f_2$ with $T$, $T_e$ is face-disjoint from $T_1$, $T_1$ does not contain vertex $w_1$ and $T_e$ does not contain vertex $u$ of $T$. If there is a \triplet $T_{f_1}$ that contains $f_1$ such that  \begin{inparaenum}[(i)]
\item $T_{f_1}$ shares only vertices of $f_1$ with $T$, and, 
\item $T_{f_1}$ is face-disjoint from all of $T_1$, $T_2$ and $T_e$,\end{inparaenum} then $\mathcal{T}=\{T\}$ is \bad. 
\end{restatable}

\begin{proof}[Sketch]
Assuming that $\mathcal{T}$ is \notbad, we arrive at the configuration shown in Fig.~\ref{fig:shortProofs:c1c2NoVertex} where the position of $z$ is not fixed. Edges $(u,z)$ and $(w_1,z)$ belong to $T_{f_1}$ but at least one has three crossings, e.g. edge $(u,z)$ in Fig.~\ref{fig:shortProofs:c1c2NoVertex}.
\end{proof}

For the second subcase, where $T_1=T_e$, 
Property~\ref{obs:twoK5sShareOnly2Vertices} assures that $T_1$ and $T_2$ share only vertices $v_2$ and $w_2$, as otherwise
$\mathcal{T}'=\{T_1=T_e,T_2\}$ would be \bad. 
As indicated by the next lemma, no further restrictions (imposed by the assumption that $\mathcal{T}'$ is \notbad) are needed to 
prove that $\mathcal{T}=\{T\}$ is \bad.

\begin{restatable}[$\mathbf{\ast}$]{lemma}{doubleKFive}\label{lem:noDoubleK5}
Let $\mathcal{T}=\{T\}$. Assume that $T_1$ is face-disjoint from $T$, that $T_2$ shares faces $f$ and $f_2$ with $T$, and $T_e=T_1$. If $T_2$ has exactly two common vertices with $T_1$, then $\mathcal{T}$ is \bad.
\end{restatable}
\begin{proof}[Sketch]
Assuming that $\mathcal{T}$ is \notbad,  we obtain the configuration in Fig.~\ref{fig:shortProofs:doubleK5}. Here, two \facialKFives have three common vertices  contradicting Property~\ref{obs:twoK5sShareOnly2Vertices}.
\end{proof}

In the following two lemmas, we summarize our findings from this section.

\begin{restatable}[$\mathbf{\ast}$]{lemma}{noTwoTriplets}\label{lem:no2triplets}
Let $G$ be optimal $2$-planar and let $\mathcal{T} = \{T\}$. $\mathcal{T}$ is \notbad if and only if every set $\mathcal{T}'=\{T_1,T_2\}$ is \bad.
\end{restatable}
\begin{proof}[Sketch]
As $G$ is optimal $2$-planar, at least one of $\mathcal{T}$ or $\mathcal{T}'=\{T_1,T_2\}$ for~some \triplets $T_1$ and $T_2$ is \notbad. If $\mathcal{T}$ is \bad, the lemma holds.
For the reverse direction, assume there is a set $\mathcal{T}'=\{T_1,T_2\}$ that is \notbad, i.e.  
 none of Lemmas~\ref{lem:c1c3_conditions}, \ref{lem:c1c2TeContainsVertex_conditions} and \ref{lem:c1c2TeNoVertex_conditions} applies for $\mathcal{T}' \cup \{T_e\}$. If $\mathcal{T}$ is \notbad from the corresponding Lemmas~\ref{lem:c1c3}, \ref{lem:c1c2TeContainsVertex} and \ref{lem:c1c2TeNoVertex}, the conditions of Lemma~\ref{lem:noDoubleK5} are satisfied and $\mathcal{T}$ is \bad.
\end{proof}

\begin{restatable}[$\mathbf{\ast}$]{lemma}{lemmasForTriplet}\label{lem:lemmasForTriplet}
Let $G$ be optimal $2$-planar and let $\mathcal{T} = \{T\}$. $\mathcal{T}$ is \bad if and only if the conditions of at least one of Lemmas~\ref{lem:c1c3}, \ref{lem:c1c2TeContainsVertex}, \ref{lem:c1c2TeNoVertex}
and~\ref{lem:noDoubleK5} are met.
\end{restatable}
\begin{proof}[Sketch]
For a contradiction we assume that $\mathcal{T}$ is bad.
Since $G$ is optimal $2$-planar, by Lemma~\ref{lem:no2triplets},  there exists a set $\mathcal{T}'=\{T_1,T_2\}$ that is \notbad. We go through all lemmas and check their conditions. For all the cases, we argue that if the specific lemma can not  apply on $\mathcal{T}$, then $\mathcal{T}'$ is \bad.
\end{proof}

\section{The Recognition and Embedding Algorithm}
\label{sec:algo}

Now that we have all required ingredients, we are ready to state our main theorem for the recognition of optimal $2$-planar graphs.

\begin{theorem}\label{thm:recognitionAlgorithm}
Let $G$ be a simple graph on $n$ vertices. It can be decided in $\Oh(n)$ time whether $G$ is optimal $2$-planar. If the instance is positive, an \textemb of $G$ is reported.
\end{theorem}

The remainder of this section contains the proof of Theorem~\ref{thm:recognitionAlgorithm} which is split into two parts. In Section~\ref{subsec:algo} we describe our algorithm and prove its correctness, while an efficient implementation is discussed in Section~\ref{subsec:time}. 

\subsection{Recognition Algorithm }\label{subsec:algo}
Our recognition algorithm is formalized in Algorithm~\ref{algo:main}.
\begin{algorithm}[t]\caption{Recognition of Simple Optimal $2$-planar Graphs}
\label{algo:main}
\KwIn{A $9$-degenerate graph $G=(V,E)$ with $|E|=5|V|-10$}
\DontPrintSemicolon
\SetArgSty{normal}
\KwOut{An \textemb of $G$ if it exists, \\
otherwise, \texttt{false}}
\BlankLine
Classify each edge $e \in E$ as potentially planar or \crossing\;
Identify all instances of \crossingConfiguration in $G$ and fix their partial rotation schemes\;
Create subgraph $G_p$\;
\If{$G_p$ is not $3$-connected or is not planar}
{
    \Return \texttt{false}\;
}
Compute the unique planar rotation scheme $\emb_p$ of $G_p$\;
\If{$\emb_p$ contains a face  of length greater than $5$ \\
\textbf{or if} a face does not induce a complete subgraph in $G$}
{
    \Return \texttt{false}\;
}
Augment $G_p$ to maximal planar by triangulating the faces of  $\emb_p$\;
Compute the set of \triplets of $\emb_p$\;
\For{every \triplet $T$ of $\emb_p$}
{
    Label $T$ as \facialKFive or as non-\facialKFive\;
}
Let $\mathcal{T}^*$ be the set of \triplets labelled as \facialKFives\;
\If{$\mathcal{T}^*$ covers all faces of $\emb_p$ and all \crossing edges of $G$  exactly once}
{
     \Return
     {rotation scheme obtained by making each $T \in \mathcal{T}^\ast$ a \facialKFive}\;
}
\Return \texttt{false}\;

\end{algorithm}
There are four main steps in the process. The first step is the classification of all edges of $G$ as potentially planar and \crossing (line 1). Second, is the identification of the \crossingConfiguration instances and the creation of $G_p$ and its planar rotation scheme $\emb_p$ (lines 2--10). Third, we identify all \triplets of $\emb_p$ (line 11).
Finally, we decide which of the \triplets are \notbad (lines 12--13) determining an \textemb of $G$ if one exists.
Having the set of \triplets that are \facialKFives in any \compliant \textemb of $G$, we decide if $G$ is optimal $2$-planar (lines 14--15)  and compute an \textemb of $G$ if it exists (line 16). The details of the steps are given in \arxapp{Theorem~\ref{lem:algo_correct} in Appendix~\ref{app:algo}}{\cite{arxivVersion}}.

\subsection{Implementation}\label{subsec:time}
First, we check that the input graph has at most $5n-10$ edges and that it is $9$-degenerate. This process takes $\Oh(n)$ time and is described in \arxapp{Details~\ref{app:prp:basic_implementation} in Appendix~\ref{app:algo}}{\cite{arxivVersion}}, together with standard techniques and data structures that we use. 
In the following, we discuss in more detail the time complexity of the involved steps of Algorithm~\ref{algo:main}. 
The first step of the process, described in line 1 of Algorithm~\ref{algo:main},  can be performed in $\Oh(n)$ time as stated in the following lemma:

\begin{restatable}[$\mathbf{\ast}$]{lemma}{classifyEdgesTime}\label{lem:classifyEdgesTime}
All edges of $G$ can be classified as potentially planar or \crossing in $\Oh(n)$ time.
\end{restatable}

We proceed with the second step given in lines 2--10 of Algorithm~\ref{algo:main}. First, we  compute the \crossingConfiguration instances of $G$ (line 2).

\begin{restatable}[$\mathbf{\ast}$]{lemma}{crossingConfigurationRecognize}\label{lem:crossingConfigurationRecognize}
All  \crossingConfiguration instances of $G$ can be identified in $\Oh(n)$ time if all edges are already classified as \crossing or potentially planar. 
\end{restatable}

For each \crossingConfiguration instance, we choose the  
rotation scheme of Fig.~\ref{fig:crossingConfigurationEmbed1}~where $(u,y)$ crosses $(u',y')$ and $(u,x)$ crosses $(u',x')$ by reclassifying   $(u,y)$ and $(u,x)$ as \crossing.
We proceed with lines 3--10 of Algorithm~\ref{algo:main}. We check the necessary conditions for $G_p$ and compute the planar rotation scheme $\emb_p$ and the dual $G_p^*$ after augmenting $G_p$ to maximal planar; see \arxapp{Details~\ref{app:planar_gp} in Appendix~\ref{app:algo}}{\cite{arxivVersion}}. 

Next, we compute the set of \triplets of $\emb_p$ (line 11) in $\Oh(n)$ time; see \arxapp{Details~\ref{app:tripletIdentification} in Appendix~\ref{app:algo}}{\cite{arxivVersion}}. During this process, for each face and each \crossing edge of $G$, we store the set of \triplets that contain it.
The next lemma describes how to implement the labeling of each \triplet as \facialKFive or not; see line 13.

\begin{restatable}[$\mathbf{\ast}$]{lemma}{tripletDecisionAlgorithm}\label{lem:tripletDecisionAlgorithm}
Let $\mathcal{T} = \{T\}$ and $\mathcal{S}_1$,  $\mathcal{S}_2$  and $\mathcal{S}_e$ be the set of triplets containing  $e_1$, $e_2$ and $e$, respectively. It can be decided in $\Oh(1)$ time if $\mathcal{T}$ is \bad or not. 
\end{restatable}

After labeling all \triplets as \facialKFives or non-\facialKFives, we consider the set $\mathcal{T}^*$ of \triplets labelled as \facialKFives. 
Checking whether  $\mathcal{T}^*$ 
covers all faces and \crossing edges of $G$ exactly once, in line 15 of Algorithm~\ref{algo:main},  can be done in $\Oh(|\mathcal{T}^*|)=\Oh(n)$ time. 
Finally, we augment the planar rotation scheme $\emb_p$ of $G_p$ to an \textemb of $G$ in $\Oh(n)$ time by inserting the \crossing edges of $G$ and identifying the crossings; refer to \arxapp{Details~\ref{app:finishEmb} in Appendix~\ref{app:algo}}{\cite{arxivVersion}}.
The overall time required for the last step  is in $\Oh(n)$.

Every step described above takes time $\Oh(n)$, hence Algorithm~\ref{algo:main} can be implemented to run  in linear time.
We point out that after fixing $\emb_p$, the \textemb of $G$ is uniquely defined if it exists. In particular, 
% \begin{remark}
if there are $c$ instances of \crossingConfiguration, there exist at most $2^c$ \textembs which can be easily enumerated.  
% \end{remark}

\section{Conclusions}\label{sec:conclusions}

We showed that simple optimal $2$-planar graphs can be recognized and embedded in linear time.
It remains open to extend our result also for simple $2$-planar graphs with the maximum number of edges when  $n\not \equiv 2\mod 3$, as well as to non-simple optimal $2$-planar graphs.
Another reasonable attempt would be  recognizing $5$-maps, similarly to~\cite{DBLP:journals/algorithmica/Brandenburg19}.  
Finally,
a natural question would be if a recognition algorithm for non-simple optimal $2$-planar graphs could be adopted for optimal $3$-planar graphs that are non-simple, or if our approach could work for simple optimal $3$-planar~graphs. 

\paragraph{Acknowledgement.}
We thank Michael Bekos for his valuable suggestions and endless discussions on this topic.

\bibliographystyle{splncs04}
\bibliography{ref}

\newpage

\arxapp{
\appendix
\input{appendix}}
{}
\end{document}

%% file: crossBat.tex
In this section, we study  the 3-connected spanning subgraph $G_p$ 
formed by the potentially planar edges of $G$. We want to compute a rotation scheme of $G_p$ which is extendable to an \textemb of $G$, if it exists. If 
in each \textemb of $G$ subgraph $G_p$ is plane, then~the~rotation scheme of $G_p$ is unique and easy to compute.
Though we cannot assure this~property, we prove  that in any \textemb of $G$, crossings between edges of $G_p$ occur in restricted configurations:
A \crossingConfiguration instance is an~induced subgraph $H$ of $G$ isomorphic to the graph $G_{CB}$ shown in Fig.~\ref{fig:crossingConfigurationStart}, so that \begin{inparaenum}[(i)] 
\item for~edge $(u_H,u'_H)$ of $H$  isomorphic to  $(u,u')$ of $G_{CB}$, it holds that  $V(H)=N(u_H)\cup N(u'_H)$, and 
\item the isomorphism between $H$ and $G_{CB}$ preserves the classification of edges to \crossing or potentially planar.
\end{inparaenum}
   
In particular, \crossingConfiguration has ten vertices, named as in Fig.~\ref{fig:crossingConfigurationStart}. Vertices $u$ and $u'$ have degree $9$ in $G$ and they
have $8$ common neighbors (the remaining vertices of \crossingConfiguration). The edges of \crossingConfiguration form four $5$-cliques that pairwise share two vertices, and are shown as \facialKFives in Fig.~\ref{fig:crossingConfigurationStart}. We call  edge $(u,u')$  the \emph{base edge} of \crossingConfiguration.
No other edges with both endpoints in \crossingConfiguration exist, except possibly for edges $(v,v')$ and  $(w,w')$. 
Fig.~\ref{fig:crossingConfigurationStart} shows the classification of edges of \crossingConfiguration as potentially planar and \crossing. 

Let $\emb$ be an \textemb of $G$, such that two potentially planar edges cross each other. These two edges must belong to the same \facialKFive and, in particular,  also to  an instance of \crossingConfiguration.

\begin{figure}[t]
    \centering
    \myhfil 
    \begin{subfigure}{0.3\textwidth}
    \centering
    \includegraphics[page=3]{figures/special_subgraph_new.pdf}
    \caption{}
    \label{fig:crossingConfigurationStart}
    \end{subfigure}
    \myhfil
    \begin{subfigure}{0.3\textwidth}
    \centering
    \includegraphics[page=2]{figures/special_subgraph_new.pdf}
    \caption{}
    \label{fig:crossingConfigurationEmbed1}
    \end{subfigure}
    \myhfil 
    \begin{subfigure}{0.3\textwidth}
    \centering
    \includegraphics[page=1]{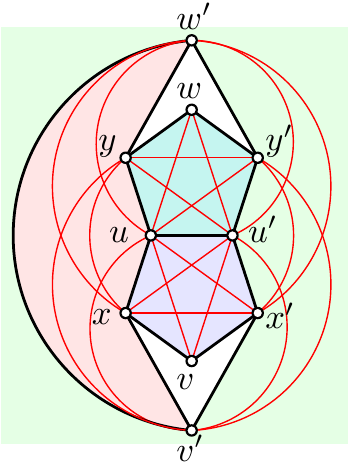}
    \caption{}
    \label{fig:crossingConfigurationEmbed2}
    \end{subfigure}
    \caption{
    (a)~Graph $G_{CB}$: red and black edges are  \crossing  and potentially planar, resp., blue edges might be both, dashed edges may be absent. (b)--(c)  the two possible rotation schemes: 
    in (b) edges $(u,x')$, $(u,y')$, $(u',x)$ and $(u',y)$ are crossing-free; in (c) edges $(u,x)$, $(u,y)$, $(u',x')$ and $(u',y')$ are crossing-free.} 
    \label{fig:crossBatCases}
\end{figure}

\begin{restatable}[$\mathbf{\ast}$]{lemma}{crossBat}\label{lem:crossBat}
Let $\emb$ be an \textemb of $G$ and let $C=\langle c_0,\ldots,c_4 \rangle$ be a facial $5$-clique in $\emb$ such that $(c_1,c_3)$ and $(c_2,c_4)$ are potentially planar edges.
Then, vertices $c_2$ and $c_3$ have degree $9$ in $G$, and the induced subgraph $H$ of $G$ with vertex set $V(H)=N(c_2)\cup N(c_3)$
is an instance of \crossingConfiguration  where $C$ is the $5$-clique $\langle 
v,x',u',u,x
\rangle$ of  Fig.~\ref{fig:crossingConfigurationStart}. 
\end{restatable}

Next, we show that  \crossingConfiguration has only two rotation schemes:

\begin{restatable}[$\mathbf{\ast}$]{lemma}{crossBatEmb}\label{lem:crossBatEmb}
Any instance of \crossingConfiguration in an \textemb $\emb$ has one of the two rotation schemes shown in Figs.~\ref{fig:crossingConfigurationEmbed1} and~\ref{fig:crossingConfigurationEmbed2}. 
\end{restatable}

\begin{proof}[Sketch]
Vertices are named as in 
Fig.~\ref{fig:crossingConfigurationStart}. First, we prove that there are four \facialKFives, namely \begin{inparaenum}[(i)]
\item $C_v$ that contains $v$,
\item $C_w$ that contains $w$,
\item $C_{x,y}$ that contains both $x$ and $y$, and,
\item $C_{x',y'}$ that contains both $x'$ and $y'$.
\end{inparaenum}
It is then easy to argue that $x,x'\in C_v$, 
while $y,y'\in C_w$.
Finally, we  show that $w$ is contained in the crossing-free cycle $(u_1,u_2,y',w',y)$ and  $v$~is contained in the crossing-free cycle $(u_2,u_1,x,v',x')$ where $\{u_1,u_2\}=\{u,u'\}$. 
The two choices for $u_1$ and $u_2$ give 
the two different rotation schemes of Figs.~\ref{fig:crossingConfigurationEmbed1} and \ref{fig:crossingConfigurationEmbed2}.
\end{proof}

As it is evident in Figs.~\ref{fig:crossingConfigurationEmbed1} and \ref{fig:crossingConfigurationEmbed2}, if an optimal $2$-planar graph $G$ contains an instance of \crossingConfiguration as subgraph, then it admits two different optimal $2$-planar rotation schemes $\emb$ and $\emb'$, that 
only differ in the choice of the rotation scheme of \crossingConfiguration. Hence for any instance of \crossingConfiguration, we may arbitrarily choose one of its two rotation schemes. Next, we formalize this observation:

\begin{lemma}\label{lem:crossingConfigurationChooseEmbed}
Let $\emb$ be an optimal 2-planar rotation scheme of $G$, and let $H$ be an instance of \crossingConfiguration in $G$. If $H$ has the rotation scheme of Fig.~\ref{fig:crossingConfigurationEmbed1} in $\emb$, then there exists another optimal 2-planar rotation scheme $\emb'$ of $G$, in which $H$ has the rotation scheme of Fig.~\ref{fig:crossingConfigurationEmbed2}, and vice-versa.
\end{lemma}

 Both rotation schemes of a \crossingConfiguration instance contain two crossings between  potentially planar edges:

\begin{restatable}[$\mathbf{\ast}$]{lemma}{crossBatTrue}\label{lem:crossBatTrue}
Let $\emb$ be an \textemb of $G$.
If $G$ contains an instance $H$ of \crossingConfiguration, then there exist exactly two pairs of potentially planar edges that belong to $H$ and  cross in $\emb$.
\end{restatable}

Using Lemma~\ref{lem:crossBatEmb} we can fix a rotation scheme for each instance $H$ of \crossingConfiguration identified in $G$  
by reclassifying two crossing potentially planar edges of $H$ to  \crossing edges. After performing this 	reclassification for all instances, if $G$ is optimal $2$-planar, then the subgraph $G_p$ induced by the potentially planar edges is planar. Furthermore, in this case, Lemma~\ref{lem:crossingConfigurationChooseEmbed} guarantees that the unique planar embedding of $G_p$ is part of an \textemb of $G$.

%% file: appendix.tex
\section{Ommited Proofs from Section~\ref{sec:preliminaries}}

\crossedEdgesShareAKFive*
\begin{proof}
Assume for a contradiction that $(u,v)$ and $(u,w)$ belong to two different {\facialKFive}s $C$ and $C'$ of an \textemb $\emb$ of $G$, respectively.
Since  $C\subseteq N(u,v)$ and $C'\subseteq N(u,w)$, it follows that $C\cup C'\subseteq S$ and thus $|C\cup C'|<8$. 
Since $|C|=|C'|=5$, it follows that $|C \cap C'| > 2$; contradiction to  Property~\ref{obs:twoK5sShareOnly2Vertices}. 
Hence, $(u,v)$ and $(u,w)$ belong to the same \facialKFive.
\end{proof}

\section{Ommited Proofs from Section~\ref{sec:crossBat}}
\label{app:crossBat}

\crossBat*

\input{crossbatproof}

\crossBatEmb*

\setcounter{section}{3}
\setcounter{lemma}{2}
\renewcommand\thesection{\arabic{section}}
\begin{proof}
Consider an instance $H$ of \crossingConfiguration  in $\emb$. In the following, we use the naming of vertices according to Fig.~\ref{fig:crossingConfigurationStart} for the vertices of $H$. First we prove that the four $5$-cliques of $H$ are always \facialKFives of $\emb$.
\begin{restatable}{proposition}{pairOfCrossinigEdges}\label{prop:pairOfCrossinigEdges}
Edges $(u,v)$ and $(u',v)$ belong to the same \facialKFive $C_{v}$ in $\emb$ with $V(C_v)\setminus \{u,u',v\} \subset \{x,x',v'\}$. Also, edges $(u,w)$ and $(u',w)$ belong to another \facialKFive $C_{w}$ with $V(C_w)\setminus \{u,u',w\} \subset \{y,y',w'\}$.
\end{restatable}

\begin{proof}
By definition, vertices $u$ and $u'$ have degree 9 in $G$ and $V(H)=N(u)\cup N(v)$. In particular $N(u) = V(H) \setminus \{u\}$ and $N(u') = V(H) \setminus \{u'\}$. 
On the other hand, the neighborhood of $v$ in $H$ consists of $u$, $u'$, $x$ and $x'$ (which form a $5$-clique with $v$) and potentially $v'$ (if the edge $(v,v')$ exists). We conclude that $N(u,v)\subseteq\{u',x,x',v'\}$, and $N(u',v)\subseteq\{u,x,x',v'\}$. Then $N(u,v) \cup N(u',v)\setminus\{u,u'\}\subseteq \{x,x',v'\}$. By Corollary~\ref{lem:crossedEdgesShareAK5}, it follows that $(u,v)$ and $(u',v)$ belong to the same \facialKFive $C_v$. The vertex set of $C_v$ contains vertices $u$, $u'$ and $v$, and two vertices from $\{x,x',v'\}$. By symmetry, the same holds for edges $(u,w)$ and $(u',w)$, that is, they belong to the same \facialKFive $C_w$, whose vertex set contains vertices $u$, $u'$, $w$, and two vertices from $\{y,y',w'\}$.
\end{proof}

\begin{restatable}{proposition}{cxy}\label{prop:cxy}
Let $\{u_1,u_2\}=\{u,u'\}$. 
Edges $(y,x)$, $(y,v')$ and $(x,w')$ belong to a \facialKFive $C_{x,y}=\langle u_1,y,w',v',x\rangle$ in $\emb$. Also, edges $(y',x')$ and $(y',v')$ and $(x',w')$ belong to another \facialKFive $C_{x',y'}=\langle u_2,x',v',w',y'\rangle$.
\end{restatable}

\begin{proof}
Consider edge $(y,x)$. Since $(y,x)$ is \crossing it belongs to a \facialKFive  $C_{x,y}$ of $\emb$. We have that $N(y,x)=\{u,u',v',w'\}$. By Proposition~\ref{prop:pairOfCrossinigEdges}, the edge $(u,u')$ belongs to two \facialKFives, that contain vertices $v$ and $w$ respectively. Since $N(y,x)$  contains neither $v$ nor $w$, it follows that $C_{x,y}$ contains exactly one of $u$ or $u'$. Then $V(C_{x,y})=\{x,y,u_1,v',w'\}$, where $u_1\in \{u,u'\}$. Hence  all three edges $(y,x)$ and $(y,v')$ and $(x,w')$ belong to the same \facialKFive, namely $C_{x,y}$.
By symmetry, edges $(y',x')$ and $(y',v')$ and $(x',w')$ belong to another \facialKFive $C_{x',y'}$ where $V(C_{x',y'})=\{x',y',u_2,v',w'\}$ and $u_2\in \{u,u'\}$.
It remains to argue that $u_1\neq u_2$. Recall that vertices $u$ and $u'$ have degree nine in $G$ and therefore belong to three \facialKFives in $\emb$. Already, by Proposition~\ref{prop:pairOfCrossinigEdges} we have identified two \facialKFives. Clearly if $u_1=u_2=u$, vertex $u$ would belong to two more \facialKFives, namely $C_{x,y}$ and $C_{x',y'}$, that is a total of four \facialKFives.
\end{proof}
As initially mentioned, we identified four \facialKFives, $C_v$, $C_w$, $C_{x,y}$ and $C_{x',y'}$.  
Recall that $C_v$ contains two vertices from $\{x,x',v'\}$. We show now, that $C_v$ does not contain $v'$. Assume for a contradiction that $C_v$ contains $v'$ and w.l.o.g. vertex $x$. Then $V(C_v) \cap V(C_{x,y})= \{u_1,x,v'\}$; a contradiction to Property~\ref{obs:twoK5sShareOnly2Vertices}. Hence, $C_v$ contains both $x$ and $x'$. Analogously, $C_w$ contains both of $y$ and $y'$.

By Proposition~\ref{prop:cxy},  $C_{y,x}$ and $C_{y',x'}$ are two \facialKFives that share the crossing-free edge $(v',w')$. Moreover, by Property~\ref{lem:crossingEdgesAreNotShared}, edge $(u_1,u_2)=(u,u')$ is crossing-free  as it belongs to the two \facialKFives $C_v$ and $C_w$ by Proposition~\ref{prop:pairOfCrossinigEdges}; refer to Figs.~\ref{fig:crossingConfigurationEmbed1} and \ref{fig:crossingConfigurationEmbed2}. Hence, $w$ is contained in the crossing-free  separating cycle $(u_1,u_2,y',w',y)$ while  $v$ is contained in the crossing-free separating cycle $(u_2,u_1,x,v',x')$. Since $(u_1,w)$ and $(u_2,w)$ are \crossing, edges $(w,y)$ and $(w,y')$ are on the boundary of $C_w$, the same holds for  edges $(v,x)$ and $(v,x')$ being on the boundary of $C_v$.  
Now, by selecting $u_1=u$ or $u_1=u'$, we yield the two different embeddings of Figs.~\ref{fig:crossingConfigurationEmbed1} and \ref{fig:crossingConfigurationEmbed2}.
\end{proof}

\setcounter{section}{2}
\renewcommand\thesection{\Alph{section}}

\crossBatTrue*

\begin{proof}
By Lemma~\ref{lem:crossingConfigurationChooseEmbed}, we can assume w.l.o.g.\ that $H$ has the rotation scheme of \crossingConfiguration shown in Fig.~\ref{fig:crossingConfigurationEmbed1} in $\emb$. Then, it is not hard to see that edges $(u',x')$ and $(u,x)$ are potentially planar and cross each other. The same holds for edges $(u',y')$ and $(u,y)$. It is evident that no other pair of potentially planar edges that belong to $H$ cross.
\end{proof}

\section{Ommited Proofs from Section~\ref{sec:triplets}}
\label{app:triplets}

\oneTripletFaceDisjoint*

\begin{proof}
Assume for a contradiction that $\mathcal{T}'$ is \notbad. In $\emb_{\mathcal{T}'}$, edge $e_1$ is assigned to $T_1$ and edge $e_2$ to $T_2$. By the assumptions of the lemma, none of $T_1$ and $T_2$ is face-disjoint from $T$ (i.e., Case C.\ref{c:1} applies neither for $T_1$ nor for $T_2$).
We consider two cases, depending on whether Case C.\ref{c:2} applies for one of $T_1$ and $T_2$ or Case C.\ref{c:3} applies for both $T_1$ and $T_2$.

First, consider the case where w.l.o.g.\ $T_1$ 
shares faces $f$ and $f_1$ with $T$ (i.e., Case C.\ref{c:2} applies for $T_1$); the reasoning for $T_2$ is symmetric. We claim that $T_2$ is face-disjoint from $T$, contradicting our assumption that none of $T_1$ and $T_2$ is face-disjoint from $T$. Indeed, $T_2$ cannot share faces $f$ and $f_2$ with $T$, as otherwise $e_1$ and $e_2$ cross each other in $\emb_{\mathcal{T}'}$ and thus $T_1$ and $T_2$ would be identical with $T$. On the other hand, if $T_2$ contains face $f_1$ (which is also contained in $T_1$), then it follows that $T_1$, $T_2$ and $T$ coincide; a contradiction. Hence, $T_2$ is necessarily face-disjoint with $T$, contradicting our assumption that Case C.\ref{c:1} does not apply for $T_2$.

Second, consider the case where Case~C.\ref{c:3} applies for both $T_1$ and $T_2$, that is, for $\{i,j\}=\{1,2\}$ \triplet 
$T_i$ shares only face $f_j$ with $T$. Clearly $T_1\neq T_2$.
In particular, in $\emb_{\mathcal{T}'}$ the \facialKFive formed by $T_1$ contains vertices $u,v_2,w_2$ (the three vertices of face $f_2$) and vertex $v_1$ (which is endpoint of edge $e_1$). Similarly, in $\emb_{\mathcal{T}'}$ the \facialKFive formed by $T_2$ contains vertices $u,v_1,w_1,v_2$. Hence, the two \facialKFives share three vertices; a contradiction to Property~\ref{obs:twoK5sShareOnly2Vertices}. This concludes the proof.
\end{proof}

\notBothTripletsFaceDisjoint*

\begin{proof}
Assume for a contradiction that $\mathcal{T}'$ is \notbad, and each of $T_1$ and $T_2$ is face-disjoint from $T$. Recall that in $\emb_{\mathcal{T}'}$ edge $e_1$ is assigned to $T_1$ and edge $e_2$ to $T_2$, where $T_1=T_2$ might hold.
Since $\mathcal{T}'$ is \notbad, $\emb_{\mathcal{T}'}$ can be extended to an \compliant \textemb $\emb'$ of $G$. This implies that there exists (at least) one \triplet $T_f$ that contains face $f$, such that  $\mathcal{T}'\cup\{T_f\}$ is \notbad.
If $T_f$ contains $f_1$, then  by Observation~\ref{obs:crossBothFacesOrNone} edge $e_1$ is assigned to $T_f$, i.e., $T_f=T_1$. This contradicts our assumption that $T_1$ is face-disjoint with $T$. Hence, $T_f$ does not contain $f_1$. Similarly, we argue that $T_f$ does not contain $f_2$.

By the assumptions of the lemma, $T_1$ and $T_2$ are face-disjoint from $T$. 
Thus, in  $\emb_{\mathcal{T}'}$ both edges $e_1$ and $e_2$ are drawn outside the faces of $T$, which implies that they cross each other.
Let $u,w_2,x_2,x_1$ and $w_1$ be the vertices of $T_f$ as shown in Fig.~\ref{fig:triplet_proof1}.  Now edge $(w_1,w_2)$ is a crossing edge of the \facialKFive formed by $T_f$ in $\emb_{\mathcal{T}'}$, which, by Corollary~\ref{col:edgesIncidentToCrossedEdgeDoNotCross}, implies that $e_1$ and $e_2$ do not cross each other; a contradiction. 
\end{proof}

\begin{figure}[t]
    \centering
    \includegraphics[page=2]{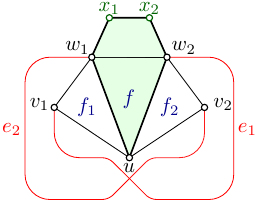}
    \caption{Illustration for the proof of~Lemma~\ref{lem:notBothTripletsFaceDisjoint}.}
    \label{fig:triplet_proof1}
\end{figure}

\cOneCThreeConditions*

\begin{proof}
Assume that $\mathcal{T}'$ is \notbad. Then $\partialemb[']$ can be extended to  a \compliant \textemb $\emb'$ of $G$. Consider the \facialKFives of $T_1$ and $T_2$ in $\partialemb[']$. Since the boundary edges of \facialKFives are crossing-free and $T_2$ shares only face $f_1$ with $T$, edge $e_2$ must cross edge $(v_1,u)$; see Fig.~\ref{fig:triplet_proof2}. 
$T_2$ contains vertices $v_1$, $w_1$, $u$ (since these are vertices of $f_1$) and vertex $v_2$ (as an endpoint of edge $e_2$). Let $z$ be the last vertex of $T_2$. If $z\in V(T)$ then $T_2=T$ by Observation~\ref{obs:eitherOnlyTwoVerticesOrFace}, contradicting the assumption that $T_2$ shares only face $f_1$ with $T$. Since edges $e$ and $e_2$ are both incident to $v_2$ and are crossing edges in $\partialemb[']$, edges $(v_2,z)$ and $(v_2,u)$ are crossing-free. Similarly, edges $(v_1,z)$ and $(v_1,w_1)$ are crossing-free. Hence, the \facialKFive of $T_2$ is 
$\langle w_1,v_1,z,v_2,u \rangle$;  see Fig~\ref{fig:triplet_proof2}. 

\begin{figure}[t]
    \centering
    \includegraphics[page=3]{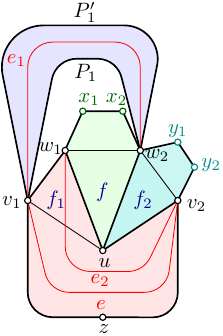}
    
    \caption{Illustration for the proof of Lemma~\ref{lem:c1c3_conditions}.}
    \label{fig:triplet_proof2}
\end{figure}

Let $T_f$ and $T_3$ be the \triplets that faces $f$ and $f_2$ are assigned to w.r.t. $\emb'$. It must be the case that $T_f\neq T_3$ as otherwise, $T_f$ would contain faces $f$ and $f_2$ of $T$ and, as a consequence, edge $e_2$. Then $T_f$ would be identical to $T_2$ contradicting the fact that $T_2$ shares only face $f_1$ with $T$.
Hence we conclude that there exist \triplets $T_f$ and $T_3$ such that $\mathcal{T}'\cup\{T_f, T_3\}$ is \notbad. Observe that $u$ is surrounded by the three \facialKFives formed by $T_2$, $T_f$ and $T_3$ in $\emb'$, that pair-wise share an uncrossed edge incident to $u$.  This implies that $d(u)=9$. Hence, if $d(u) > 9$, then $\mathcal{T}'$ is \bad as stated in the lemma. It remains to prove that there is a vertex $x \in S=N(u) \setminus V(T)$ with $|N(x)\cap S|< 2$.

The \facialKFives formed by \triplets  $T_f$ and $T_3$ share vertices $u$ and $w_2$, thus, by Property~\ref{lem:crossingEdgesAreNotShared}, edge  $(u,w_2)$ is crossing-free in $\emb'$. Similarly, 
$T_2$ and $T_3$ share vertices $u$ and $v_2$,
while $T_2$ and $T_{f}$ share vertices $u$ and $w_1$. Hence, edges $(u,v_2)$ and $(u,w_1)$ are also crossing-free in $\emb'$.

Thus, the \facialKFive formed by $T_f$ is bounded by the crossing-free $5$-cycle $(u,w_2,x_2,x_1,w_1)$ while the \facialKFive formed by $T_f$ is bounded by the crossing-free $5$-cycle $(u,v_2,y_2,y_1,w_2)$ where $x_1,x_2,y_1,y_2 \not \in V(T)$; see Fig~\ref{fig:triplet_proof2}.
Observe that since $T_{f}$ and $T_3$ share vertices $u$ and $w_2$ we can conclude that $\{x_1,x_2\} \cap \{y_1,y_2\} = \emptyset$ by Property~\ref{obs:twoK5sShareOnly2Vertices}. Similarly,  each of $T_{f}$ and $T_3$  shares two vertices with $T_2$, thus $z \not \in \{x_1,x_2,y_1,y_2\}$. 

Since \triplets $T_1$ and $T_f$ share vertex $w_2$,  at least one vertex among $x_1$ and $x_2$, say $x_1$, does not belong to $T_1$ by Property~\ref{obs:twoK5sShareOnly2Vertices}. Let $S=N(u) \setminus V(T)=\{x_1,x_2,y_1,y_2,z\}$. We claim that $|N(x_1) \cap S| \leq 1$.
Since $T_1$ is a \facialKFive in $\partialemb[']$, there exist two crossing-free paths $P_1$ and $P_1'$ that connect $v_1$ and $w_2$. Assume w.l.o.g. that the edges of $P_1$ and $P_1'$ incident to $v_1$ appear in this order and between edges $(v_1,w_1)$ and $(v_1,z)$ in the circular order of edges around $v_1$\footnote{It might be that the edge of $P_1$ incident to $v_1$
coincides with  $(v_1,w_1)$ or that the edge of $P_1'$ incident to $v_1$
coincides with  $(v_1,z)$.}; refer to Fig.~\ref{fig:triplet_proof2}. The cycle $(w_2,P_1,v_1,w_1,u)$ is crossing-free in $\emb'$. Since $x_1$ lies in the interior of this cycle, while vertices $z$, $y_1$ and $y_2$ lie in its exterior, 
it follows that vertex $x_1$ is not adjacent to vertices $z$, $y_1$ or $y_2$, i.e., $x_1$ is adjacent only to vertex $x_2$ of $S$.

To summarize, assuming that $\mathcal{T}'$ is \notbad, vertex $u$ has degree nine and there exists at least one neighbor of $u$ that does not belong to $V(T)$, namely $x_1$, such that $|N(x_1)\cap S|\leq 1$. The lemma follows.
\end{proof}

\cOneCThree*

\begin{proof}
Assume for a contradiction that $\mathcal{T}$ is \notbad. Let $\emb$ be a \compliant \textemb of $G$ that extends $\partialemb$. In $\emb$, the \triplet $T_x$ that contains edge $(v,x)$ contains three more neighbors of $v$.  Since $|N(x)\cap S|\leq 1$,  $T_x$ contains at least two vertices from $V(T)\setminus \{v\}$, contradicting Property~\ref{obs:twoK5sShareOnly2Vertices}.
\end{proof}

\containsVertexConditions*

\begin{proof}
Assume that $\mathcal{T}'$ is \notbad and that $T_1$ contains vertex $w_1$. Then $\partialemb[']$ can be extended to  a \compliant \textemb $\emb'$ of $G$. 
Since $e_1$ belongs to $T_1$, it follows that $T_1$ also contains vertices $v_1$ and $w_2$ of $T$. Let $x_1$ and $x_2$ be the other two vertices of $T_1$. First we claim that neither $x_1$ nor $x_2$ is a vertex of $T$. Assume the contrary, say for $x_1$. Then, $x_1$ is either vertex $u$ or vertex $v_2$ of $T$. If $x_1$ is vertex $u$, then by Observation~\ref{obs:eitherOnlyTwoVerticesOrFace}, \triplet $T_1$ would contain face $f_1$ of $T$ contradicting the assumption that $T_1$ is face-disjoint with $T$. On the other hand, if $x_1$ is vertex $v_2$, then $T_1$ would contain edge $e_2$, contradicting the fact that $T_1$ is different from $T_2$.

Now, consider a counter-clockwise walk around the boundary of the \facialKFive formed by $T_1$ in $G_p$. If edge $(v_1,w_1)$ is not on the boundary of $T_1$, then face $f_1$ would belong to $T_1$; a contradiction to the fact that $T_1$ is face-disjoint with $T$. Similarly, $(w_1,w_2)$ is also on the boundary of $T_1$, since face $f$ does not belong to $T_1$. This implies that vertices $v_1$, $w_1$ and $w_2$ appear consecutively along the boundary of $T_1$, followed by vertices $x_1$ and $x_2$.

\begin{figure}[t]
    \centering
    \includegraphics[scale=1,page=1]{figures/optimal_c1c2.pdf}
    \caption{Illustration for the proof of Lemma~\ref{lem:c1c2TeContainsVertex_conditions}.}
    \label{fig:c1c2:1}
\end{figure}

Recall that w.r.t. $\partialemb[']$ edges $e_1$, $e_2$ and $e$ are assigned to $T_1$, $T_2$ and $T_e$ respectively, as shown in  Fig.~\ref{fig:c1c2:1}. Note that  $f_1$ is not contained in any of $T_1$, $T_2$ or $T_e$. Let $T_3$ be the \triplet that forms the \facialKFive of $\emb'$ where $f_1$ belongs to. Since $\mathcal{T}'$ is \notbad, we conclude that $\mathcal{T}'\cup \{T_3\}$ is also \notbad.  Observe that $w_1$ is surrounded by the three \facialKFives formed by $T_1$, $T_2$ and $T_3$ in $\emb'$, that pair-wise share an uncrossed edge incident to $w_1$. 
Thus, vertex $w_1$ is a degree-9 vertex. Clearly, if $d(w_1) > 9$, $\mathcal{T}'$ is \bad as stated in the lemma. It remains to prove that there is no vertex $x \in S=N(w_1) \setminus V(T)$ with $|N(x)\cap S| \geq 4$.

By the assumptions of the lemma, \triplet $T_2$ contains faces $f$ and $f_2$ of $T$, hence it contains vertices $w_1$, $w_2$, $u$ and $v_2$ of $T$. Let $z$ be the last vertex of $T_2$, which is different from $v_1$. 
Now, along the boundary of $T_3$ vertices $u$, $w_1$ and $v_1$ appear consecutively, followed by  two more vertices $y_1,y_2 \not \in V(T)$ that are different from vertices $x_1$, $x_2$ and $z$. 
Hence the neighbors of $w_1$ in circular order  are  
$\{v_1,y_1,y_2,u,z,v_2,w_2,x_1,x_2\}$ or $\{v_1,y_1,y_2,u,v_2,z,w_2,x_1,x_2\}$ depending on whether vertex $z$ appears before of after vertex $v_2$.
Let $S=\{x_1,x_2,y_1,y_2,z\}$ be the set of neighbors of $w_1$ that do not belong to $T$. We claim that every vertex of $S$ has at most $3$ neighbors in $S$.

Recall that the boundary edges of the \facialKFives in $\partialemb['\cup\{T_3\}]$ are crossing-free. Since $T_e$ is a \facialKFive in $\partialemb['\cup\{T_3\}]$, it follows that there exist two crossing-free paths $P_e$ and $P_e'$ that connect vertices $v_1$ and $v_2$. 
Assume w.l.o.g. that the edges of $P_e'$ and $P_e$ incident to $v_1$ appear in this order and between edges $(v_1,x_2)$ and $(v_1,y_1)$ in the circular order of edges around $v_1$\footnote{It might be that the edge of $P_e'$ incident to $v_1$
coincides with  $(v_1,x_2)$ or that the edge of $P_e$ incident to $v_1$
coincides with  $(v_1,y_1)$.}; refer to Fig.~\ref{fig:c1c2:1}. 
Then $P_e$ along with the crossing-free path $v_2,u,w_1,v_1$ forms a crossing-free cycle in $\partialemb['\cup\{T_3\}]$. Vertex $y_2$ lies in the interior of this cycle (note that $y_1$ might be a vertex of this cycle), while vertex $x_1$ lies in its exterior. Thus, $x_1$ and $y_2$ can't be adjacent vertices. Also, any other vertex of $S$ can be adjacent to at most one of $x_1$ and $y_2$, proving the claim.
To summarize, assuming that $\mathcal{T}'$ is \notbad, and that vertex $w_1$ has degree nine, we conclude that for every neighbor $x$ of $w_1$ in $S$ it holds $|N(x)\cup S|\leq 3$, proving the first part of the lemma. 
The case where $\mathcal{T}'$ is \notbad  and $T_e$ contains vertex $u$ is symmetric.
\end{proof}

\containsVertex*

\begin{proof}Assume for a contradiction that $\mathcal{T}$ is \notbad. Thus there is a \compliant \textemb $\emb$ of $G$ that extends $\partialemb$.
In $\emb$,  the neighbors of $v$ belong to  three \facialKFives, one of them being $T$. Let $T'$ and $T''$ be two other \facialKFives which include $v$ so that  $T'$ and $T''$  share vertices $v$ and $c \neq v$. Clearly, $c\in S$, and  $|N(v) \cap N(c)|\geq 6$. Since each of $T'$ and $T''$ can contain at most one vertex of $T$  different from $v$, this implies that $c$ must be adjacent to at least $4$ vertices of $S$; a contradiction.
\end{proof}

\noVertexConditions*
\begin{proof}
Assume that $\mathcal{T}'$ is \notbad, i.e. there exists a \compliant \textemb $\emb'$ of $G$ that extends  $\partialemb[']$. By assumption, $T_2$ shares only faces $f$ and $f_2$ with $T$ in $\emb'$. In $\emb'$, face $f_1$ is part of a \facialKFive formed by a \triplet $T_{f_1}$ which is necessarily face-disjoint from all of $T_1$, $T_2$ and $T_e$. Since $f_1$ already shares vertices $u$ and $w_1$ with the \facialKFive formed by $T_2$, \triplet $T_{f_1}$ contains neither vertex $v_2$ nor $w_2$ of $T$, i.e. $T_{f_1}$ shares only the three vertices of $f_1$ with $T$. The lemma follows.
\end{proof}

\noVertex*

\begin{proof}
Assume for a contradiction that $\mathcal{T}$ is \notbad.
The boundary of \triplet $T_1$ consists of two paths in $\emb_p$, say $P_1$ and $P'_1$, that connect vertices $v_1$ and $w_2$. Similarly, the boundary of $T_e$ consists of two paths in $\emb_p$, say $P_e$ and $P'_e$, that connect vertices $v_1$ and $v_2$. 
Note that $P_1$ and $P'_1$ are edge-disjoint. The same holds for paths 
$P_e$ and $P'_e$. Assume w.l.o.g. that in the counter-clockwise order of edges around $v_1$, we encounter the edges of $P_1$, $P_1'$, $P_e'$ and $P_e$ incident to $v_1$ in this order between edges $(v_1,w_1)$ and $(v_1,u)$; see Fig.~\ref{fig:c1c2:part2}. Thus, $P_1$ and $P_e$ are edge-disjoint. On the other hand, paths 
$P'_1$ and $P'_e$ might either be edge-disjoint, as in Fig.~\ref{fig:c1c2:4},  or they share an edge, as shown in Figs.~\ref{fig:c1c2:5} and~\ref{fig:c1c2:8}. Note that in the second case the edge of $P'_1$ incident to $v_1$ is the same as the corresponding edge of $P'_e$.

\begin{figure}[t]
\centering
    \begin{subfigure}[b]{0.31\textwidth}
    \includegraphics[page=5]{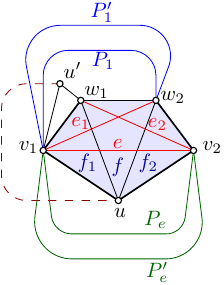}
    \caption{}
    \label{fig:c1c2:4}
    \end{subfigure}
    \myhfil
    \begin{subfigure}[b]{0.335\textwidth}
    \includegraphics[page=8]{figures/optimal_c1c2.pdf}
    \caption{}
    \label{fig:c1c2:5}
    \end{subfigure}
    \myhfil
    \begin{subfigure}[b]{0.335\textwidth}
    \includegraphics[page=6]{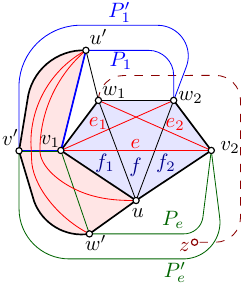}
    \caption{}
    \label{fig:c1c2:8}
    \end{subfigure}
    \caption{Illustration for the proof of Lemma~\ref{lem:c1c2TeNoVertex}.}
    \label{fig:c1c2:part2}
\end{figure}

Consider \triplet $T_{f_1}$ of the lemma that contains face $f_1$  of $T$. Since by assumption $T_{f_1}\neq T_2$, \triplet  $T_{f_1}$ must contain at least one of the triangular faces that share edge $(v_1,w_1)$ or $(v_1,u)$ with face $f_1$. 
Assume w.l.o.g. that $T_{f_1}$ contains the face sharing edge $(v_1,w_1)$ with  $f_1$; the case where $T_{f_1}$ contains the face sharing edge $(v_1,u)$ with $f_1$ is symmetric. Let $u'$ be the third vertex of this face that is different from vertices $v_1$ and $w_1$; see Fig.~\ref{fig:c1c2:part2}. Since both $u$ and $u'$ are vertices of \triplet $T_{f_1}$, edge $(u,u')$ is an edge of $G$. Furthermore $u'$ lies in the interior of the cycle $(w_2,P_1',v_1,w_1)$. In particular, the presence of path $P'_1$ in $G_p$ implies that edge $(u,u')$ is a \crossing edge of $G$. 

Since $\mathcal{T}$ is \notbad there exists a \compliant \textemb $\emb$ of $G$ that extends $\partialemb$.
We focus on edge $(u,u')$ in $\emb$, and consider two cases, depending on whether paths $P'_1$ and $P'_e$ are edge-disjoint or not. If $P'_1$ and $P'_e$ are edge-disjoint, the edge $(u,u')$ in $\emb$ must cross an edge of $P'_1$, an edge of $P_e$ and an edge of $P'_e$, as shown in Fig.~\ref{fig:c1c2:4}. This implies that $\emb$ cannot exist, since $(u,u')$ has at least three crossings; a contradiction to the assumption that $\mathcal{T}$ is \notbad. If $P'_1$ and $P'_e$ share edge $(v_1,v')$, then in $\emb$, edge $(u,u')$ must cross the edge of $P_e$ that is incident to $v_1$, say $(v_1, w')$,  and the common edge $(v_1,v')$ of paths $P'_1$ and $P'_e$ (otherwise, we can conclude as before that $\mathcal{T}$ is \bad). Furthermore, vertex $u'$ belongs to path $P_1$ 
as otherwise $(u,u')$ also crosses an edge of $P_1$ creating three crossings along $(u,u')$. This implies that in $\emb$,  $\langle v_1, u', v', w', u \rangle$ is a \facialKFive, and edge $(u',w')$ is a crossing edge of this \facialKFive; see Figs.~\ref{fig:c1c2:5} and~\ref{fig:c1c2:8}. 

Now, \triplet $T_{f_1}$ contains vertices $u$, $v_1$, $w_1$ and $u'$. Let $z$ be the last  vertex of $T_{f_1}$ which by the assumptions of the lemma is not a vertex of $T$, and it is adjacent to vertices $u$, $w_1$ and $v_1$. 

Since by the assumption of the lemma, $T_{f_1}$ is face-disjoint from $T_1$ and $T_e$, it follows that $z$ is in the interior of the region bounded by cycle $(w_2,P'_1,v_1,w_1)$ or the region bounded by cycle $(v_1,P'_e,v_2,u)$. In the first case; see Fig.~\ref{fig:c1c2:5}, edge $(u,z)$ can't be drawn since it should cross paths $P_1'$, $P_e'$ and $P_e$, having at least three crossings; a contradiction. Similarly, in the second case; see Fig.~\ref{fig:c1c2:8}, edge $(w_1,z)$ can't be drawn since it should cross paths $P_1$, $P_1'$ and $P_e'$, having at least three crossings; a contradiction. The lemma follows.
\end{proof}

\doubleKFive*

\begin{figure}[t]
\centering
\begin{subfigure}[b]{0.3\textwidth}
    \includegraphics[width=\textwidth,page=1]{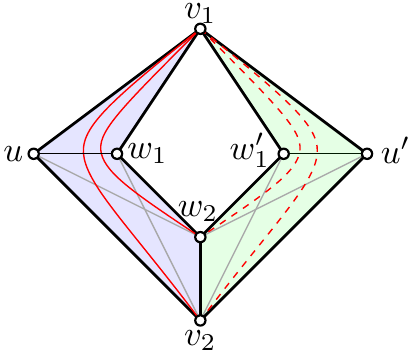}
    \caption{}
    \label{fig:doubleKfive_1}
    \end{subfigure}
     \myhfil
     \begin{subfigure}[b]{0.3\textwidth}
    \includegraphics[width=\textwidth,page=3]{figures/doubleKfive.pdf}
    \caption{}
    \label{fig:doubleKfive_3}
    \end{subfigure}
    \caption{Illustration for the proof of Lemma~\ref{lem:noDoubleK5}.}
    \label{fig:doubleKfive}
    
\end{figure}
\begin{proof}
From the assumptions stated in the lemma, edges $e$ and $e_1$ belong to both \triplets $T$ and $T_1$. This implies that $T$ and $T_1$ share vertices $v_1$, $w_2$ and $v_2$. If $T$ and $T_1$ share another vertex, then $T_2$ and $T_1$ would have at least three common vertices, contradicting the assumptions of the lemma.

Hence, $T$ and $T_1$ form the configuration shown in Fig.~\ref{fig:doubleKfive_1} which is unique up to a renaming of vertices $w_1$ and $u$ and a renaming of vertices $w_1'$ and $u'$. Triplet $T_2$ contains vertices $w_1$, $w_2$, $u$, $v_2$ and a vertex $z$ which by the assumptions of the lemma does not belong to $T$ or $T_1$.

Assume that $\mathcal{T}$ is not bad, i.e., there exists a \compliant \textemb $\emb$ of $G$ that extends $\partialemb$ such that $T$ is a \facialKFive in $\emb$. Assume w.l.o.g. that $z$ lies in the exterior of the cycle $Z=(v_1,u,v_2,u')$ of $G_p$ as shown in Fig.~\ref{fig:doubleKfive_3}. Since $T_2$ is a \triplet, edges $(z,w_1)$  and $(z,w_2)$ belong to $G$. Furthermore, $(z,w_1)$ and $(z,w_2)$ are \crossing edges, as $w_1$ and $w_2$ lie in the interior of $Z$ and $z$ in its exterior; see Fig.~\ref{fig:doubleKfive_3}. 

In $\emb$, edge $(z,w_1)$ must cross both edges $(v_1,u')$ and $(v_1,w'_1)$. This implies that $C=\langle z, v_1, w_1, w_1', u\rangle$ is a \facialKFive in $\emb$; see Fig.~\ref{fig:doubleKfive_3}.
Similarly, edge $(z,w_2)$ must cross edges $(v_2,u')$ and $(v_2,w_1')$ (the latter crossing is unavoidable in the presence of the \facialKFives formed by $T$ and $C$), hence $C'=\langle z, u', w_1', w_2, v_2 \rangle$ is another \facialKFive of $\emb$. Then $C$ and $C'$ form two \facialKFives in $\emb$ and share vertices $z$, $u'$ and $w'_1$, contradicting  Property~\ref{obs:twoK5sShareOnly2Vertices}.  
\end{proof}

\noTwoTriplets*

\begin{proof}
Since $G$ is optimal $2$-planar, then at least one of $\mathcal{T}$ or $\mathcal{T}'=\{T_1,T_2\}$ for some \triplets $T_1$ and $T_2$ is \notbad. Hence if $\mathcal{T}$ is \bad, the lemma holds. It remains to argue that if $\mathcal{T}$ is \notbad, then every set $\mathcal{T}'=\{T_1,T_2\}$ is \bad.
Assume that this is not the case, and there exist \triplets $T_1$ and $T_2$ such that $\mathcal{T}'=\{T_1,T_2\}$ is \notbad. 

By Lemmas~\ref{lem:oneTripletFaceDisjoint} and~\ref{lem:notBothTripletsFaceDisjoint}, exactly one of $T_1$ and $T_2$ is face-disjoint with $T$. Assume, w.l.o.g. that $T_1$ is face-disjoint with $T$, i.e. Case~C.\ref{c:1} applies for $T_1$. Then either $T_2$ shares only face $f_1$ with $T$ or it shares only faces $f$ and $f_2$ with $T$, i.e either Case~C.\ref{c:3} or Case~C.\ref{c:2} applies for $T_2$, respectively. Thus, $T_1\neq T_2$ also holds.

In the first case, by Lemma~\ref{lem:c1c3_conditions} it follows that vertex $u$ of $T$ has degree $9$ and that there exists a vertex $x\in S=N(u)\setminus V(T)$ such that $|N(x)\cap S|\leq 1$. 
On the other hand, applying  Lemma~\ref{lem:c1c3}  for vertex $u\in V(T)$, 
we conclude that $\mathcal{T}$ is \bad, contradicting our assumption that both $\mathcal{T}$ and  $\mathcal{T}'$ are \notbad.
Hence, we can assume that $T_2$ shares only faces $f$ and $f_2$ with $T$, i.e. Case~C.\ref{c:2} applies for $T_2$. Consider edge $e$ of $T$. Since $\mathcal{T}'$ is \notbad, $\partialemb[']$ can be extended to at least one \compliant \textemb  of $G$. Let $\emb'$ be such a rotation scheme and let \triplet $T_e$ be the \triplet that edge $e$ is assigned to w.r.t. $\emb'$. We observe first that $T_e\neq T_2$. Indeed, since $T_2$ shares vertices $u$, $w_1$, $w_2$ and $v_2$ with $T$, $T_e=T_2$ would imply that $T_2=T$; a contradiction. 
We distinguish two cases, depending on whether  $T_e \neq T_1$ or $T_e= T_1$.

\myparagraph{Case 1: $T_e \neq T_1$.}
Since $T_e\neq T_1$  and $T_e\neq T_2$ it follows that $e$ belongs neither to $T_1$ nor to $T_2$. Furthermore, the existence of $\emb'$ implies that  $\mathcal{T}''=\mathcal{T}'\cup \{T_e\}$ is \notbad. Applying Lemma~\ref{lem:c1c2TeContainsVertex} for $\mathcal{T}$, we conclude that for every vertex $v\in V(T)$ with degree nine, there exists a vertex $x\in S=N(v)\setminus V(T)$ such that $|N(x)\cap S|\geq 4$.
In particular, this property holds for vertices $w_1$ and $u$ if any of them has degree nine. Since $\mathcal{T}''$ is \notbad, we conclude from Lemma~\ref{lem:c1c2TeContainsVertex_conditions} that $T_1$ does not contain $w_1$ while $T_e$ does not contain $u$.
Hence,  by Lemma~\ref{lem:c1c2TeNoVertex_conditions} for  $\mathcal{T}''$, there exists a \triplet $T_{f_1}$ that contains face $f_1$ of $T$, shares only vertices of $f_1$ with $T$, and is face-disjoint from all \triplets of $\mathcal{T}''$.
Then, Lemma~\ref{lem:c1c2TeNoVertex} implies that $\mathcal{T}$ is \bad; a contradiction.

\myparagraph{Case 2: $T_e = T_1$.}
Both \triplets $T$ and $T_1$ contain edges $e_1$ and $e$, hence they contain their endpoints, i.e. $T$ and $T_1$ share at least three vertices. Since $\mathcal{T}'$ is \notbad, and  since \triplets $T_1$ and $T_2$ both contain vertices $v_2$ and $w_2$,   \triplets $T_1$ and $T_2$ have exactly two common vertices by Property~\ref{obs:twoK5sShareOnly2Vertices}. Now the conditions of Lemma~\ref{lem:noDoubleK5} are fulfilled and we conclude that $\mathcal{T}$ is \bad; a contradiction. This completes the proof of the lemma.
\end{proof}

\lemmasForTriplet*
\begin{proof}
Clearly, if we can apply at least one of these lemmas for $\mathcal{T}$ then $\mathcal{T}$ is \bad. So, assume that $\mathcal{T}$ is bad, and that for each of these lemmas not all conditions hold for $\mathcal{T}$. Since $G$ is optimal $2$-planar, by Lemma~\ref{lem:no2triplets},  there exists a set $\mathcal{T}'=\{T_1,T_2\}$ that is \notbad. 

By Lemmas~\ref{lem:oneTripletFaceDisjoint} and~\ref{lem:notBothTripletsFaceDisjoint}, exactly one of $T_1$ and $T_2$ is face-disjoint with $T$. Assume, w.l.o.g. that $T_1$ is face-disjoint with $T$. Then either $T_2$ shares only face $f_1$ with $T$ or it shares only faces $f$ and $f_2$ with $T$.

In the first case, since $\mathcal{T}$ does not fulfill all conditions of Lemma~\ref{lem:c1c3}, we conclude that, in particular for vertex $u\in V(T)$,  $d(u)>9$ holds or every vertex $x\in S=N(u)\setminus V(T)$ has $|N(x)\cap S|\geq 2$. Then, the conditions of Lemma~\ref{lem:c1c3_conditions} are satisfied, and  $\mathcal{T}'$  is \bad; a contradiction.

Hence, we can assume that $T_2$ shares only faces $f$ and $f_2$ with $T$. Consider edge $e$ of $T$. Since $\mathcal{T}'$ is \notbad, $\partialemb[']$ can be extended to at least one \compliant \textemb  of $G$. Let $\emb'$ be such an embedding and let \triplet $T_e$ be the \triplet that  $e$ is assigned to w.r.t. $\emb'$. 

Since $\mathcal{T}$ does not satisfy all conditions of Lemma~\ref{lem:noDoubleK5}, either $T_e\neq T_1$ or $T_2$ does not have exactly two common vertices with $T_1$. If $T_2$ has at least three common vertices with $T_1$, then by Property~\ref{obs:twoK5sShareOnly2Vertices}, set $\mathcal{T}'$ is \bad.
On the other hand, if $T_2$ has exactly one common vertex with $T_1$, then $T_e\neq T_1$, as otherwise, $T_2$ and $T_1$ would share vertices $w_2$ and $v_2$. Hence, it suffices to consider the case where $T_e\neq T_1$.
Note that $T_e\neq T_2$ also holds. Indeed, since $T_2$ shares vertices $u$, $w_1$, $w_2$ and $v_2$ with $T$, $T_e=T_2$ would imply that $T_2=T$ also holds; a contradiction. 

Then $\mathcal{T}''=\{T_1,T_2,T_e\}$ is \notbad, which implies that $T_e$ is face-disjoint from $T_1$ and $T_2$.
Since $\mathcal{T}$ does not satisfy all conditions of Lemma~\ref{lem:c1c2TeContainsVertex}, in particular for vertex $w_1\in V(T)$, $d(w_1)>9$ holds or there exists a vertex $x\in S=N(w_1)\setminus V(T)$ such that $|N(x)\cap S|\geq 4$. If $T_1$ contains vertex $w_1$, then Lemma~\ref{lem:c1c2TeContainsVertex_conditions} implies that $\mathcal{T}''$ is \bad; a contradiction. Hence, $T_1$ does not contain vertex $w_1$. Using the same argument for vertex $u\in V(T)$, we conclude that $T_e$ does not contain vertex $u$. 
Already all conditions of Lemma~\ref{lem:c1c2TeNoVertex} are satisfied except for the constraints related to $T_{f_1}$.
Since by assumption not all conditions of Lemma~\ref{lem:c1c2TeNoVertex} are met, we conclude that every \triplet $T_{f_1}$ that contains $f_1$
\begin{inparaenum}[(i)]
\item shares at least four vertices with $T$, or, \item
is not face-disjoint with at least one \triplet of $\mathcal{T}''$.\end{inparaenum} Then all conditions of  Lemma~\ref{lem:c1c2TeContainsVertex_conditions} are satisfied and $\mathcal{T}''$ is \bad.
The lemma follows.
\end{proof}

\section{Ommited Proofs and Details from Section~\ref{sec:algo}}
\label{app:algo}

\begin{theorem}\label{lem:algo_correct}
Algorithm~\ref{algo:main} is correct.
\end{theorem}

\begin{proof}
We assume that the input graph $G$ has exactly $5|V|-10$ edges and is $9$-degenerate. We proceed with the four steps of our algorithm as described in Section~\ref{subsec:algo} and formalized in Algorithm~\ref{algo:main}.
In the first step, we use Corollary~\ref{col:identificationProperty} to classify each edge of $G$ as \crossing or potentially planar. Recall that an edge is potentially planar if and only if its two endpoints have at least six common neighbors.

From Lemma~\ref{lem:crossBat} we get that if two potentially planar edges of $G$ cross each other in some \textemb, then they belong to an instance of \crossingConfiguration. Hence, in the second step, we identify all such instances and by Lemma~\ref{lem:crossBatEmb} we can fix the partial rotation scheme of the subgraph induced by the vertices of each \crossingConfiguration instance. By doing this, we change the classification of two potentially planar edges of each instance to \crossing edges.
Note that by Lemma~\ref{lem:crossBatTrue} for every \crossingConfiguration instance, there exist exactly two pairs of potentially planar edges that cross, hence the 		reclassification process takes place only when potentially planar edges cross in some \textemb.

As a result, the subgraph $G_p$ defined by the set of potentially planar edges is planar. Furthermore, if $G$ is optimal 2-planar, then  $G_p$ must also be 3-connected and spanning $G$, as it contains as subgraph the 3-connected pentangulation of every \compliant \textemb of $G$.
Hence if $G_p$ is not planar or if $G_p$ is not 3-connected we can conclude that $G$ is not optimal-2-planar. If $G_p$ is planar and 3-connected, we compute its unique planar rotation scheme $\emb_p$.
Note that by Lemma~\ref{lem:crossingConfigurationChooseEmbed}, if $G$ is optimal 2-planar, then there exists a \compliant \textemb $\emb$ of $G$. As mentioned at the beginning of Section~\ref{sec:triplets}, all faces of $\emb_p$ must have length $3$, $4$ or $5$. In particular non-triangular faces belong to \facialKFives of $\emb$ and hence they induce complete subgraphs in $G$. If these properties are violated, we conclude that $\emb$ does not exist and $G$ is not optimal 2-planar. 
In order to proceed with the next step, we first need  to augment $G_p$ to maximal planar. To achieve this, we arbitrarily triangulate every non-triangular face of $\emb_p$ (this process can not create parallel edges or self-loops in $G_p$ as it is 3-connected). Also, since non-triangular faces induce complete subgraphs in $G$, the added edges belong to $G$ and are classified as \crossing. We change their classification to potentially planar so that they belong to $G_p$. 

The next step is to calculate the \triplets of $\emb_p$, which can be accomplished by computing the paths of length two of the dual $G^*_p$ of $G_p$ and check whether the vertices of the three faces induce a $5$-clique in $G$. 

Now, for each \triplet $T$ we can use Lemma~\ref{lem:lemmasForTriplet} to decide whether $\mathcal{T}=\{T\}$ is \bad or not. If $\mathcal{T}$ is \notbad, we mark \triplet $T$ as a \facialKFive. By definition, if $T$ is a \facialKFive of any \compliant \textemb of $G$, then the three \crossing edges of $T$ are assigned to $T$ and to no other \triplet by  Lemma~\ref{lem:no2triplets}.
Hence if any \crossing edge of $G$ is not assigned to exactly one \triplet, then we can conclude that $G$ is not optimal-2-planar. Similarly, a face of $\emb_p$ is assigned to exactly one \triplet. If this property does not hold, then $G$ is not optimal 2-planar. 

Having reached this point, $G$ is an optimal 2-planar graph and it remains to compute its \compliant \textemb by combining the embedding $\emb_p$ of $G_p$ and the \triplets marked as \facialKFives. To do this, for each \triplet $T=\langle f_1,f,f_2\rangle$ as depicted in Fig.~\ref{fig:triplet_base}, for every vertex $v$ of $T$ we insert the \crossing edges of $T$ that are incident to $v$ between the corresponding edges of $G_p$ in the cyclic order of edges around $v$. Also, we define the order of crossings along each edge of $T$ as indicated by Fig.~\ref{fig:triplet_base}.

The correctness of Algorithm~\ref{algo:main} follows from the above analysis.
\end{proof}

\begin{details}
\label{app:prp:basic_implementation}In our algorithm, we assume that the given graph has $5n-10$ edges and is $9$-degenerate. Checking whether a graph has these properties can be done in linear time by using standard algorithms~\cite{DBLP:journals/jacm/MatulaB83}.

For checking whether $G_p$ is planar and 3-connected, and for computing the planar rotation scheme $\emb_p$ of $G_p$ (the cyclic order of the edges around each vertex, and the faces corresponding to the embedding), as well as the dual graph $G^*_p$ we use standard algorithms that run in linear time and data structures that consume linear space~\cite{DBLP:journals/siamcomp/HopcroftT73,DBLP:journals/jacm/HopcroftT74}.

In all steps, it is required to check whether two vertices are adjacent or to find the edge between two given vertices. 
Since $G$ is $9$-degenerate, for each vertex $v_i$ in the degeneracy sequence $v_1,\ldots,v_n$, we separately store its edges in subgraphs $G_i=G[v_1,\ldots,v_i]$ and $G_i'=G[v_i,\ldots,v_n]$. For $j > i$, edge $(v_i,v_j)$ belongs to $G_j$. Hence, it suffices to only check the stored edges of $v_j$ in $G_j$, which are at most nine. Hence, 
the cost of these operations is $\Oh(1)$.
\end{details}

\classifyEdgesTime*

\begin{proof}
First, we compute a 9-degenerate sequence $L=v_1,\ldots,v_n$ of the vertices of $G$ in linear time~\cite{DBLP:journals/jacm/MatulaB83}.
 In order to classify the edges, 
 for each edge $(u,v)$ we count the number of common neighbors of $u$ and $v$ as follows.
For $i=n,\ldots,1$ we process vertex $v_i$. 
Recall that since $G$ is 9-degenerate, vertex $v_i$ has at most nine neighbors in $G_i$.
For each of the at most ${9\choose 2}=36$ pairs of neighbors $u,w$ of $v_i$ in $G_i$, we check whether they are connected by an edge in $\Oh(1)$ time. If so, we update the number of common neighbors for the edges $(u,v_i)$, $(v_i,w)$ and $(u,w)$. 
After this process, we check whether the number of common neighbors computed for each edge $(u,v)$ is at least $6$ or not. If this is the case 
we mark $(u,v)$ as potentially planar, otherwise as \crossing.
Since for every vertex, we process at most $36$ pairs of neighbors in $\Oh(1)$ time each, the classification of all edges of $G$ takes $\Oh(n)$  time in total.
\end{proof}

\crossingConfigurationRecognize*
\begin{proof}
We iterate over all edges of $G$ trying to identify first the base edge of each \crossingConfiguration instance. Recall that each  base edge is potentially planar and both its endpoints have degree $9$. So, in our iteration we explicitly focus on such edges. Let $(u,u')$ be such an edge. If $(u,u')$ is the base edge of a \crossingConfiguration instance, then $|N(u,u')|=8$.  Since $u$ and $u'$ are degree-9 vertices, we can check in $\Oh(1)$ time if $N(u,u')$ has eight elements. If this is not true, then $(u,u')$ is not the base edge of a \crossingConfiguration instance. Otherwise, let $H$ be the subgraph induced by $N(u,u')\cup \{u,u'\}$ and let $N(u,u')=\{u_1,\ldots, u_8\}$. In the following, we try to find an isomorphism between $u_1,\ldots, u_8$ and the vertices of \crossingConfiguration as named in Fig.~\ref{fig:crossingConfigurationStart}.

Clearly, if $H$ is an instance of \crossingConfiguration, then $\{u_1,\ldots, u_8\}$ must have the same degree-sequence in decreasing order as \crossingConfiguration. If this is not true, then $H$ is not an instance of \crossingConfiguration.  

Since vertices $v$ and $w$ of \crossingConfiguration are the only vertices with $4$ or $5$ neighbors in \crossingConfiguration (depending on whether they are adjacent to $v'$ and $w'$, respectively), we can find the corresponding two vertices of $H$; assume w.l.o.g. that $\{u_1,u_2\}=\{v,w\}$. From the set of the remaining vertices $S=\{x,x',y,y',v',w'\}$ of \crossingConfiguration, only $v'$ and $w'$ have five neighbors in $S$. W.l.o.g., let $u_3=v'$ and $u_4=w'$.

Now $x$ and $x'$ are connected to $v'$ with potentially planar edges, while $y$ and $y'$ are connected to $v'$ with \crossing edges. This allows to distinguish $x$ and $x'$ from $y$ and $y'$. So, let w.l.o.g. $\{u_5,u_6\}=\{x,x'\}$ and $\{u_7,u_8\}=\{y,y'\}$. Also, $x$ and $x'$ are connected to $v$ and not connected to $w$. Hence, we can identify $v$ and $w$; w.l.o.g. let $u_1=v$ and $u_2=w$.

It remains to identify $x$, $x'$, $y$ and $y'$. Since the pairs $\{x,y\}$ and  $\{x',y'\}$ are symmetric in \crossingConfiguration, we arbitrarily choose $x=u_5$. Since $x$ is connected to $y$ but not to $y'$, this allows to identify $y$; w.l.o.g. let $y=u_7$. Then we can conclude that $x'=u_6$ and $y'=u_8$. After the mapping of the vertices of $H$ to the ones of \crossingConfiguration, it remains to check that $H$ is isomorphic to \crossingConfiguration, i.e., that it contains exactly the edges of \crossingConfiguration, and that the edges of $H$ have the same classification (\crossing or potentially planar) as the corresponding edges of the \crossingConfiguration.
 
If the above process fails at any step, then we can conclude that $H$ is not an instance of \crossingConfiguration. After identifying the two adjacent vertices $u$ and $u'$ with degree $9$, all steps  related to determining whether $(u,u')$ is a base edge of a \crossingConfiguration instance can be performed in $\Oh(1)$ 
time. 
\end{proof}

\begin{details}\label{app:planar_gp}
If $G$ is optimal $2$-planar, $G_p$ is 3-connected, planar and spans $G$. These properties can be checked in linear time  with algorithms that also compute the planar rotation scheme $\emb_p$ of $G_p$. Note that we ensured that $G_p$ is $3$-connected, hence, $\emb_p$ is uniquely defined (see line 6 of Algorithm~\ref{algo:main}). After computing $\emb_p$, we can iterate over all $\Oh(n)$ faces of $\emb_p$ and check for each face in $\Oh(1)$ time, if its length is at most five and, if so, check in $\Oh(1)$ time, if it induces a complete subgraph in $G$ (see lines 7--8). If any face violates these constraints, we reject the instance in line 9. 

Next, we triangulate all faces of length  four or five in line 10 of Algorithm~\ref{algo:main}. In fact, adding edges in $G_p$ corresponds to 		reclassifying some \crossing edges of $G$ as potentially planar. This process takes time $\Oh(n)$. Now that $G_p$ is maximal planar, we recompute its planar rotation scheme $\emb_p$ and its dual $G^*_p$ in linear time.
\end{details}

\begin{details}\label{app:tripletIdentification}
Recall that a \triplet $T=\langle f_1,f,f_2 \rangle$ consists of three faces $f_1$,$f$ and $f_2$ of $\emb_p$ which form the  path $(f_1,f,f_2)$ of length $2$ in the dual $G_p^*$. All vertices of $G_p^*$ have degree 3, hence identifying all paths of length two takes linear time in the order of $G_p^*$, i.e. $\Oh(|V(G_p^*)|)=\Oh(n)$ time. For each path we check in 
$\Oh(1)$ time, if the vertices of the three faces induce a $5$-clique in $G$. If this is the case, the path is a \triplet of $\emb_p$. During this process, for each face and each \crossing edge of $G$, we store the set of \triplets that contain it and count their cardinalities. Note that this information can be stored in linear space, since each \triplet contains three faces and three \crossing edges of $G$.
\end{details}

\tripletDecisionAlgorithm*

\begin{proof}
By Lemma~\ref{lem:lemmasForTriplet} it suffices to check whether the conditions of one of Lemmas~\ref{lem:c1c3}, \ref{lem:c1c2TeContainsVertex}, \ref{lem:c1c2TeNoVertex}
and~\ref{lem:noDoubleK5} are met.
For the first two lemmas, namely Lemmas~\ref{lem:c1c3} and~\ref{lem:c1c2TeContainsVertex}, for every vertex $v\in V(T)$ with $d(v)=9$, and every neighbor $x\notin V(T)$ of $v$, we compute the set $S=N(v,x)\setminus V(T)$ and its cardinality. Since $|N(v)|=9$ and $S\subseteq N(v)$, this computation requires to check a constant number of times for the existence of some edge. Hence, it takes $\Oh(1)$ time to check whether the conditions of Lemmas~\ref{lem:c1c3} and \ref{lem:c1c2TeContainsVertex} hold for $\mathcal{T}$.

The next two lemmas, namely Lemmas~\ref{lem:c1c2TeNoVertex} and~\ref{lem:noDoubleK5}, assume that $T_1$ is face-disjoint from $T$ and $T_2$ contains faces $f$ and $f_2$ of $T$. The case where $T_2$ is face-disjoint from $T$ and $T_1$ contains faces $f$ and $f_1$ of $T$ can be handled analogously.
In a brute-force approach for Lemma~\ref{lem:c1c2TeNoVertex}, one would have to consider all \triplets $T_{f_1}$ that contain face $f_1$ and  all possible sets $\mathcal{T}'=\{T_1,T_2,T_e\}$, where edges $e_1$, $e_2$ and $e$ belong to $T_1$, $T_2$ and $T_e$ respectively. Since $T_{f_1}$ contains face $f_1$ of $T$, there exist at most five ways to extend $f_1$ to a path of three faces in the dual  $G^*_p$ of $G_p$. Checking whether these paths of faces form \triplets of $\emb_p$ and whether they share only the vertices of $f_1$ with $T$ can be performed in $\Oh(1)$ time. Similarly, since $T_2$ shares faces $f$ and $f_2$ with $T$, there are at most three choices for its third face (face $f_1$ is excluded). For each choice we can identify in $\Oh(1)$ time if this is a \triplet of $\emb_p$.  Clearly if we cannot identify $T_{f_1}$ or $T_2$ then the conditions of Lemma~\ref{lem:c1c2TeNoVertex} are not met. 

Now, in this brute-force approach, for \triplets $T_1$ and $T_e$, there might exist a linear number of \triplets that contain edge $e_1$ (or $e$), thus we would have to consider a quadratic number of sets $\mathcal{T}'$. Ideally we would like to have a bounded number of \triplets that contain edge $e_1$ or $e$. In fact we can prove that if this is not the case, then we can always find \triplets $T_1$ and $T_e$ that satisfy the conditions of  Lemma~\ref{lem:c1c2TeNoVertex}, i.e. it applies for $\mathcal{T}$. In order to achieve this, we need the following:

\setcounter{section}{5}
\setcounter{lemma}{3}
\renewcommand\thesection{\arabic{section}}

\begin{proposition}\label{prp:findTriplet}
Let $T$ be a \triplet and consider a set $\mathcal{S}$ of \triplets with at least 20 elements. Then, at least one \triplet of the set is face-disjoint with $T$.
\end{proposition}
\begin{proof}
\begin{figure}[t]
    \centering
        \includegraphics[page=1]{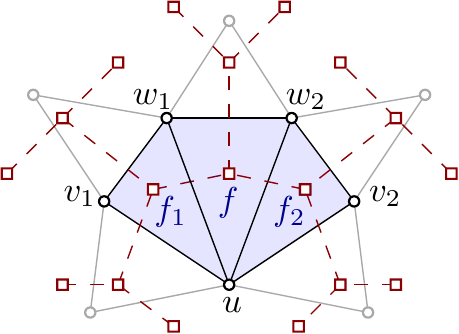}
    \caption{Illustration for the proof of Proposition~\ref{prp:findTriplet}.}
    \label{fig:findTriplets}
\end{figure}
We want to count the \triplets that contain at least one face of $T$. Each \triplet corresponds to a unique path of length $2$ in the dual $G_p^*$; see Fig.~\ref{fig:findTriplets}. 
There exist five \triplets that contain face $f_1$ but not face $f$, and four \triplets that contain both $f_1$ and $f$. Similarly there exist nine \triplets that contain face $f_2$ (four of them also include face $f$), and two \triplets that contain $f$ but not $f_1$ and $f_2$.
In total we count 20 \triplets, having counted $T$ twice: once for face $f_1$ and once for face $f_2$. Hence there exist at most  19 different \triplets that contain at least one face of $T$.
Then if the set $\mathcal{S}$ contains at least 20 elements, there is at least one \triplet that is face-disjoint with $T$.
\end{proof}

Along the same lines we can argue the following:

\begin{proposition}\label{prp:findTriplet_2}
Let $T$ be a \triplet, with vertices $u$ and $v$ that are consecutive along its boundary in $\emb_p$. Let $\mathcal{S}$ be a set of at least six \triplets that are face-disjoint from $T$ and contain vertex $u$. Then, at least one \triplet of the set does not contain vertex $v$.
\end{proposition}
\begin{proof}
Let $f$ and $f'$ be the two faces of $\emb_p$ that share edge $(u,v)$. Assume w.l.o.g. that $f$ belongs to $T$. Note that any \triplet that is face-disjoint with $T$ and contains both vertices $u$ and $v$, must contain face $f'$ (and not $f$). There exist at most five such \triplets. Hence, if $\mathcal{S}$ has at least six elements,  then it has  at least one \triplet that does not contain $f'$, i.e. does not contain vertex $v$ as claimed.
\end{proof}

Consider the set $\mathcal{S}_1$ of \triplets that contain edge $e_1$ and the set $\mathcal{S}_e$ of \triplets that contain edge $e$. Assume w.l.o.g. that $|\mathcal{S}_1|\leq |\mathcal{S}_e|$. 
If $|\mathcal{S}_e|<82$  then we can check for every $T_1\in \mathcal{S}_1$ and every $T_e\in \mathcal{S}_e$ whether the conditions of Lemma~\ref{lem:c1c2TeNoVertex} hold. Checking whether $T_1$ and $T_e$ are face-disjoint and whether they contain vertex $w_1$ and $u$, respectively, takes $\Oh(1)$ time.

So assume that $|\mathcal{S}_e|\geq 82$. 
If we can find a \triplet $T_1$ that is face-disjoint from all of $T$, $T_2$ and $T_{f_1}$, and does not contain vertex $w_1$, then we claim that Lemma~\ref{lem:c1c2TeNoVertex} applies and $\mathcal{T}$ is \bad.
Indeed, assuming a valid $T_1$, since  $|\mathcal{S}_e|\geq 82$ then Proposition~\ref{prp:findTriplet} assures that at least $63$ \triplets of $\mathcal{S}_e$ are face-disjoint from $T$. Among these \triplets, at least $44$ of them are also face-disjoint from $T_2$, at least $25$ of these $44$ are also face-disjoint from $T_1$, and at least six of them are additionally face-disjoint from $T_{f_1}$. Then by Proposition~\ref{prp:findTriplet_2} we can find a \triplet $T_e$ that, in addition to the previous properties, does not contain vertex $u$. Then all conditions of Lemma~\ref{lem:c1c2TeNoVertex} are met for $T_e$, which implies that $\mathcal{T}$ is \bad as claimed.
Hence we only need to identify \triplet $T_1$. 
If $|\mathcal{S}_1|  < 63$, then we check all elements of $\mathcal{S}_1$, whether they are face-disjoint from all of $T$, $T_2$ and $T_{f_1}$, and do not contain vertex $w_1$ in $\Oh(1)$ time. If we find \triplet $T_1$, then $\mathcal{T}$ is \bad, otherwise the conditions of Lemma~\ref{lem:c1c2TeNoVertex} are not met. 
If on the other hand,   $|\mathcal{S}_1|\geq 63$,
Propositions~\ref{prp:findTriplet} and \ref{prp:findTriplet_2} assure that there exists a \triplet $T_1$ with the aforementioned properties, hence the conditions of Lemma~\ref{lem:c1c2TeNoVertex} are satisfied.

Finally, for Lemma~\ref{lem:noDoubleK5}, first we need to find \triplets $T_1=T_e$ and $T_2$, if they exist. Note that if $T_1$ exists, then it contains the face opposite of $f_2$ along edge $(v_2,w_2)$. This face, considered as a vertex in the dual $G^*_p$ of $G_p$,  can be extended to at most five different paths of length $2$ of $G^*_p$. We can check for each path in $\Oh(1)$ time if its vertices create a $5$-clique, i.e. if this is a \triplet of $\emb_p$. On the other hand, since $T_2$ contains faces $f$ and $f_2$ of $T$, there are at most three choices for its third face (face $f_1$ is excluded). Again, we check if these options define \triplets of $\emb_p$ in $\Oh(1)$ time. 
Now, for every \triplet $T_1$ and $T_2$ we check in $\Oh(1)$ time if they are face-disjoint and whether they have exactly two common vertices. Since there exist at most five \triplets $T_1$ and at most three \triplets $T_2$, checking whether Lemma~\ref{lem:noDoubleK5} applies for $\mathcal{T}$ can be done in $\Oh(1)$ time.
\end{proof}
\setcounter{section}{4}
\renewcommand\thesection{\Alph{section}}

\begin{details}\label{app:finishEmb}
Here we augment the planar rotation scheme $\emb_p$ of $G_p$ to an \textemb of $G$ by inserting the \crossing edges of $G$ appropriately. Let \triplet $T=\langle f_1,f,f_2\rangle \in \mathcal{T}^*$ as depicted in Fig.~\ref{fig:triplet_base}. For every vertex $v$ of $T$ we insert the \crossing edges of $T$ that are incident to $v$ between the corresponding edges of $\emb_p$ in the cyclic order of edges around $v$. This can be done in $\Oh(1)$ time for the \crossing edges incident to a specific vertex $v$ in a \triplet as the \triplet consists of the faces that provide access to the correct position in the cyclic neighbor list  of $v$.
Since the \triplet also defines the crossings contained in the \facialKFive, we can define and store the order of crossings along each edge of $T$ in $\Oh(1)$ time.
\end{details}

%% file: crossbatproof.tex
\setcounter{section}{3}
\setcounter{lemma}{1}
\renewcommand\thesection{\arabic{section}}

\begin{proof}

Since $(c_1,c_3)$ is potentially planar, vertices $c_1$ and $c_3$ have at least six common neighbors. Out of these common neighbors, exactly three belong to $C$ (namely, $c_0$, $c_2$ and $c_4$). It follows that at least three of the common neighbors of $c_1$ and $c_3$, call them $u$, $v$ and $w$, do not belong to $C$; see Fig.~\ref{fig:crossingEdgeInAK5:2}. %
By Corollary~\ref{col:edgesIncidentToCrossedEdgeDoNotCross}, the edges connecting $c_1$ to $u$, $v$ and $w$ do not cross the edges that connect $c_3$ to $u$, $v$ and $w$ in $\emb$. Assume w.l.o.g.\ 
that $c_1,c_2,u,v,w,c_4,c_0$ appear in this circular order around $c_3$ in $\emb$; see Fig.~\ref{fig:crossingEdgeInAK5:2}.

\begin{figure}[t]
    \centering
    \begin{subfigure}{0.3\textwidth}
        \centering
        \includegraphics[page=1]{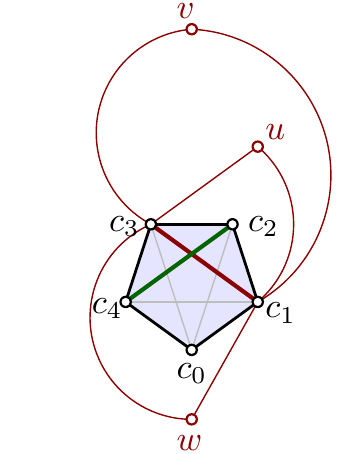}
    \caption{}
    \label{fig:crossingEdgeInAK5:2}
    \end{subfigure}
    \myhfil
    \begin{subfigure}{0.3\textwidth}
        \centering
        \includegraphics[page=2]{figures/new_crossbat}
    \caption{}
    \label{fig:crossingEdgeInAK5:3}
    \end{subfigure}
    \myhfil
    \begin{subfigure}{0.3\textwidth}
        \centering
        \includegraphics[page=3]{figures/new_crossbat}
    \caption{}
    \label{fig:crossingEdgeInAK5:4}
    \end{subfigure}

    \begin{subfigure}{0.3\textwidth}
        \centering
        \includegraphics[page=4]{figures/new_crossbat}
    \caption{}
    \label{fig:crossingEdgeInAK5:5}
    \end{subfigure}
    \myhfil
    \begin{subfigure}{0.3\textwidth}
        \centering
        \includegraphics[page=5]{figures/new_crossbat}
    \caption{}
    \label{fig:crossingEdgeInAK5:6}
    \end{subfigure}
    \myhfil
    \begin{subfigure}{0.3\textwidth}
        \centering
        \includegraphics[page=6]{figures/new_crossbat}
    \caption{}
    \label{fig:crossingEdgeInAK5:7}
    \end{subfigure}
    \caption{(a) Vertices $c_1$, $c_3$ have three common neighbors $\{u,v,w\}$ outside $C$. (b) Illustration for Proposition~\ref{prop:c1WCrossingfree}. (c) Illustration for Proposition~\ref{prop:xNorYNorZNotOutside}. (d) Illustration for Proposition~\ref{prop:wIsNotXorY}. (e) Illustration
    for Proposition~\ref{prop:wIsZ}. (f) Summary of constraints imposed by Propositions~\ref{prop:c1WCrossingfree}--\ref{prop:wIsZ}.
    }
    \label{fig:crossingEdgeInAK5_part1}
\end{figure}

A symmetric argument for $c_2$ and $c_4$ implies the existence of three common neighbors $x$, $y$ and $z$ of $c_2$ and $c_4$ that do not belong to $C$. Since $C$ is a \facialKFive, vertices $x,y,z$ appear between $c_3$ and $c_0$ in the circular order around $c_4$ in $\emb$. In particular, we can assume w.l.o.g. that vertices $c_2,c_3,x,y,z,c_0,c_1$ occur in this order around $c_4$ in $\emb$. Note that $x$, $y$ and $z$ are not necessarily different from $u$, $v$ and $w$. In particular, we will prove in the following that $|\{u,v,w\} \cap \{x,y,z\}|=2$. 

\begin{proposition}
\label{prop:c1WCrossingfree}
None of the edges $(c_4,x)$, $(c_4,y)$ and $(c_4,z)$ crosses $(c_1,w)$.
\end{proposition}

\begin{proof}
Assume that edge $(c_4,x)$ crosses $(c_1,w)$ as in Fig.~\ref{fig:crossingEdgeInAK5:3}. This contradicts Corollary~\ref{col:edgesIncidentToCrossedEdgeDoNotCross}, since $(c_1,c_4)$ is a crossing edge inside $C$ while edges $(c_4,x)$ and $(c_1,w)$ do not belong to $C$. A symmetric argument applies for edges $(c_4,y)$ and $(c_4,z)$.
\end{proof}
\begin{proposition}
\label{prop:xNorYNorZNotOutside}
None of vertices $\{x,y,z\}$ lies in the interior of the region delimited by the cycle $(c_3,w,c_1,c_0,c_4)$.
\end{proposition}
\begin{proof}
Assume for a contradiction that vertex $z$ lies in the interior of this region (refer to the gray-shaded region in  Fig.~\ref{fig:crossingEdgeInAK5:4}); an analogous argument holds for $x$ and $y$. Then, edge $(c_2,z)$ crosses at least three edges, in particular, one edge of $(c_1,u)$ and $(c_3,u)$, one edge of $(c_1,v)$ and $(c_3,v)$, and one edge of $(c_1,w)$ and $(c_3,w)$; a contradiction to the fact that $\emb$ is a \textemb.
\end{proof}

\begin{proposition}
\label{prop:wIsNotXorY}
$w \notin \{x,y\}$ and edges $(c_4,x)$ and $(c_4,y)$ cross edge $(c_3,w)$.
\end{proposition}

\begin{proof}
Assume for a contradiction, that $w = y$; an analogous argument holds if we assume $w=x$. Since $y$, $z$ and $c_0$ appear in this circular order around $c_4$ in $\emb$ and since, by Proposition~\ref{prop:c1WCrossingfree}, edge $(c_4,z)$ does not cross $(c_1,w)$, vertex $z$ is located in the region delimited by the cycle $(c_4,w=y,c_1,c_0)$; see  Fig.~\ref{fig:crossingEdgeInAK5:5}. This contradicts Proposition~\ref{prop:xNorYNorZNotOutside}, and hence $w \notin \{x,y\}$. 

Now, since $C$ is a facial $5$-clique, by Proposition~\ref{prop:xNorYNorZNotOutside}, vertices $x$ and $y$ must be in the interior of the region delimited by the cycle $(c_3,c_2,c_1,w)$. Hence, each of edges $(c_4,x)$ and $(c_4,y)$ crosses either $(c_3,w)$ or $(c_1,w)$. However, by Proposition~\ref{prop:c1WCrossingfree}, the two edges must cross $(c_3,w)$ as claimed.
\end{proof}

\begin{proposition}
\label{prop:wIsZ}
$z=w$.
 \end{proposition}
 \begin{proof}
Assume for a contradiction that $z\neq w$. By Proposition~\ref{prop:xNorYNorZNotOutside}, it follows that $z$ lies in the interior of the region delimited by the cycle $(c_3,c_2,c_1,w)$.  By Proposition~\ref{prop:c1WCrossingfree}, it follows that $(c_4,z)$ crosses edge $(c_3,w)$. However, by Proposition~\ref{prop:wIsNotXorY}, $(c_3,w)$ is also crossed by $(c_4,x)$ and $(c_4,y)$ as in Fig.~\ref{fig:crossingEdgeInAK5:6}; a contradiction to $2$-planarity. 
 \end{proof}

\begin{figure}[t]
    \centering
    \begin{subfigure}{0.3\textwidth}
        \centering
        \includegraphics[page=7]{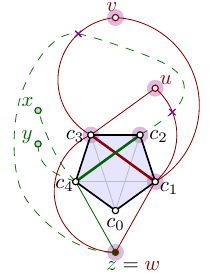}
    \caption{}
    \label{fig:crossingEdgeInAK5:8}
    \end{subfigure}
    \myhfil
    \begin{subfigure}{0.3\textwidth}
        \centering
        \includegraphics[page=8]{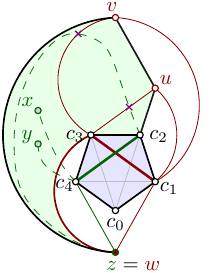}
    \caption{}
    \label{fig:crossingEdgeInAK5:9}
    \end{subfigure}
    \myhfil
    \begin{subfigure}{0.3\textwidth}
        \centering
        \includegraphics[page=9]{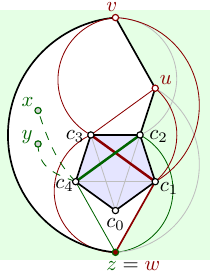}
    \caption{}
    \label{fig:crossingEdgeInAK5:10}
    \end{subfigure}
    
    \caption{(a)--(b) Illustrations for Proposition~\ref{prop:c2zFacial5Clique}. (c) Summary of constraints imposed by Propositions~\ref{prop:c1WCrossingfree}--\ref{prop:c2zFacial5Clique}.
    }
    \label{fig:crossingEdgeInAK5_part2}
\end{figure}

The constraints discussed so far lead to the configuration shown in Fig.~\ref{fig:crossingEdgeInAK5:7}.
Since $z=w$ is a common neighbor of $c_2$ and $c_4$, we turn our attention to edge $(c_2,z)=(c_2,w)$. We prove the following: 

\begin{proposition}
\label{prop:c2zFacial5Clique}
Edge $(c_2,z)$ crosses $(c_1,u)$ and $(c_1,v)$ and  $\langle c_2,c_1,z,v,u  \rangle$  is a facial $5$-clique.
\end{proposition}
\begin{proof}
By the circular order of edges around $c_1$ and $c_3$, edge $(c_2,z)$ must cross one of edges $(c_1,u)$ and $(c_3,u)$ as well as one of edges $(c_1,v)$ and $(c_3,v)$. First we prove that it must either cross $(c_1,u)$ and $(c_1,v)$ or $(c_3,u)$ and $(c_3,v)$. 
Indeed if $(c_2,z)$ crosses $(c_1,u)$ and $(c_3,v)$ as in Fig.~\ref{fig:crossingEdgeInAK5:8} (the case where  $(c_2,z)$ crosses $(c_1,v)$ and $(c_3,u)$ is symmetric), edge $(c_2,z)$ is crossed by two independent edges, contradicting Property~\ref{lem:crossingTwoEdgesBelongToSameK5}.

It remains to prove that $(c_2,z)$  does not cross edges $(c_3,u)$ and $(c_3,v)$. 
Assume the contrary; see Fig.~\ref{fig:crossingEdgeInAK5:9}. Then, by Property~\ref{lem:crossingTwoEdgesBelongToSameK5}, $\langle c_3, c_2, u, v, z \rangle$ is a \facialKFive with $(c_3,z)$ being one of its crossing-free boundary edges. However, by Proposition~\ref{prop:wIsNotXorY} edges $(c_4,x)$ and $(c_4,y)$  cross with edge $(c_3,z)$, as $w=z$; a contradiction. We conclude that  $(c_2,z)$ crosses $(c_1,u)$ and $(c_1,v)$ forming the facial $5$-clique $\langle c_2,c_1,z,v,u  \rangle$.
\end{proof}
To summarize so far, Propositions~\ref{prop:xNorYNorZNotOutside} and \ref{prop:wIsNotXorY} imply that vertices $x$ and $y$ lie in the interior of the region delimited by the cycle $(c_3,c_4,c_0,c_1,z=w)$. Since $\langle c_0,c_1,c_2,c_3,c_4  \rangle$ and $\langle c_2,c_1,z,v,u  \rangle$ are facial $5$-cliques, we can further restrict $x$ and $y$ in the interior of the region delimited by the cycle $(c_3,c_2,u,v,z)$; see Fig.~\ref{fig:crossingEdgeInAK5:10}. In the following proposition, we restrict $x$ and $y$ even more.

\begin{figure}[t]
    \centering
    \begin{subfigure}{0.3\textwidth}
        \centering
        \includegraphics[page=10]{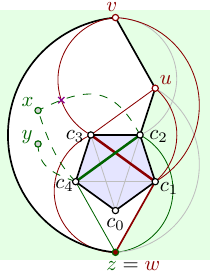}
    \caption{}
    \label{fig:crossingEdgeInAK5:11}
    \end{subfigure}
    \myhfil
    \begin{subfigure}{0.3\textwidth}
        \centering
        \includegraphics[page=11]{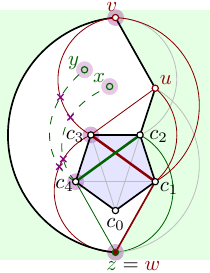}
    \caption{}
    \label{fig:crossingEdgeInAK5:12}
    \end{subfigure}
    \myhfil
    \begin{subfigure}{0.3\textwidth}
        \centering
        \includegraphics[page=12]{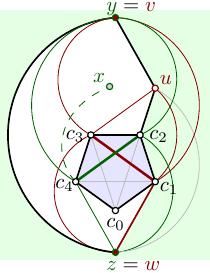}
    \caption{}
    \label{fig:crossingEdgeInAK5:13}
    \end{subfigure}

    \caption{(a)  Illustration for Proposition~\ref{prop:xNorYinC3vz}. (b) Illustration for Proposition~\ref{prop:vIsY}. (c) Summary of constraints imposed by Propositions~\ref{prop:c1WCrossingfree}--\ref{prop:vIsY}.
    }
    \label{fig:crossingEdgeInAK5_part3}
\end{figure}

\begin{proposition}
\label{prop:xNorYinC3vz}
Neither $x$ nor $y$ lies in the interior of the region delimited by the cycle $(c_3,v,z)$.
\end{proposition}
\begin{proof}
Assume to the contrary that vertex $x$ lies in the interior of the region delimited by the cycle $(c_3,v,z)$.; an analogous argument applies for $y$. Then, edge $(c_2,x)$ crosses $(c_3,v)$; see Fig.~\ref{fig:crossingEdgeInAK5:11}. This contradicts Corollary~\ref{col:edgesIncidentToCrossedEdgeDoNotCross}, since $(c_2,v)$ is a crossing edge in the facial $5$-clique $\langle c_2,c_1,z,v,u  \rangle$. 
\end{proof}
\begin{proposition}
\label{prop:vIsNotX}
$v \neq x$ and edge $(c_4,x)$ crosses $(c_3,v)$.
\end{proposition}
\begin{proof}
First assume for a contradiction that $v = x$. Due to circular order of the edges around $c_4$, vertex $y$ can only be located in the interior of the region delimited by the cycle $(c_3,v,z)$. This contradicts Proposition~\ref{prop:xNorYinC3vz}.
Hence $v \neq x$. The crossing between $(c_4,x)$ and $(c_3,v)$ is imposed by  Proposition~\ref{prop:xNorYinC3vz}.
\end{proof}
\begin{proposition}
\label{prop:vIsY}
$v=y$.
\end{proposition}
\begin{proof}
Assume for a contradiction that $v \neq y$. By Propositions~\ref{prop:xNorYinC3vz} and~\ref{prop:vIsNotX}, it follows that both edges $(c_4,x)$ and $(c_4,y)$ cross both of $(c_3,z)$ and $(c_3,v)$; see Fig.~\ref{fig:crossingEdgeInAK5:12}. Then, by Property~\ref{lem:crossingEdgesBelongToSameK5}, the six vertices $c_4,c_3,x,y,v,z$ must form a \facialKFive; a contradiction. 
\end{proof}
The constraints imposed by Propositions~\ref{prop:c1WCrossingfree}--\ref{prop:vIsY} imply the configuration shown in Fig.~\ref{fig:crossingEdgeInAK5:13}. It remains to decide whether $u=x$ or if $u$ is located in the interior of cycle $(c_3,u,y)$.
\begin{proposition}
\label{prop:uIsNotX}
$x \neq u$.
\end{proposition}
\begin{proof} 
Assume for a contradiction that $x = u$. Recall  from Proposition~\ref{prop:wIsNotXorY} that edge $(c_4,x)$  crosses edge $(c_3,z)$; see Fig.~\ref{fig:crossingEdgeInAK5:14}.  Consider now edge $(z,x)$ which is part of the facial $5$-clique $\langle c_2,c_1,z,y,x  \rangle$. By Corollary~\ref{col:edgesIncidentToCrossedEdgeDoNotCross}, edge $(c_4,x)$ is not allowed to cross edge $(c_3,z)$, which however is the case; a contradiction.
\end{proof}  

By Proposition~\ref{prop:vIsNotX} and by $2$-planarity, we conclude that  $x$ is located inside the cycle $(c_3,u,y)$; see Fig.~\ref{fig:crossingEdgeInAK5:16}. Since by Propositions~\ref{prop:wIsNotXorY} and~\ref{prop:vIsNotX}, edge $(c_4,x)$ crosses with edges $(c_3,w)$ and $(c_3,v)$, respectively, it follows by  Property~\ref{lem:crossingTwoEdgesBelongToSameK5} that $\langle c_3, x,y, z, c_4 \rangle$ is a \facialKFive.
In addition, $(c_4,z)$ is a crossing-free edge on the boundary of this facial $5$-clique.

\begin{figure}[t]
    \centering
    \begin{subfigure}{0.3\textwidth}
        \centering
        \includegraphics[page=13]{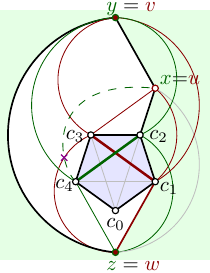}
    \caption{}
    \label{fig:crossingEdgeInAK5:14}
    \end{subfigure}
    \myhfil
    \begin{subfigure}{0.3\textwidth}
        \centering
        \includegraphics[page=14]{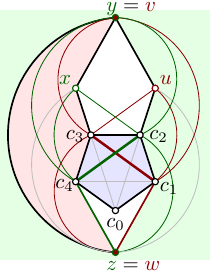}
    \caption{}
    \label{fig:crossingEdgeInAK5:16}
    \end{subfigure}
    \myhfil
    \begin{subfigure}{0.3\textwidth}
        \centering
        \includegraphics[page=15]{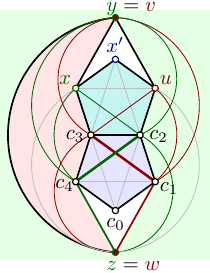}
    \caption{}
    \label{fig:crossingEdgeInAK5:18}
    \end{subfigure}

    \caption{(a) Illustration for Proposition~\ref{prop:uIsNotX}. (b) Summary of constraints imposed by Propositions~\ref{prop:c1WCrossingfree}--\ref{prop:uIsNotX}. (c) The crossing between $(c_3,u)$ and $(c_2,x)$ implies the \facialKFive $\langle x,c_3,c_2,u,x' \rangle$.
    }
    \label{fig:crossingEdgeInAK5_part4}
\end{figure}

Now, since edges $(c_2,x)$ and $(c_3,u)$ cross, they belong to a  \facialKFive that contains vertices $c_3,c_2,u,x$. Let $x'$ be the last vertex of this \facialKFive. It must be $x'\neq y$ since edge $(c_3,y)$ is a crossing edge inside the \facialKFive $\langle c_3,x,y,z,c_4\rangle$. 

The final configuration is depicted in Fig.~\ref{fig:crossingEdgeInAK5:18}. Vertices $c_2$ and $c_3$ have degree nine in $G$, with eight common neighbors $N(c_2)\cap N(c_3)=\{c_0,c_1,c_4,u,v,w,x,x'\}$. 
Let $H$ be the subgraph induced by the vertex set $V(H)=N(c_2)\cup N(c_3)$.
Since cycles $(c_0,c_4,z,c_1)$ and $(x,x',u,y)$ are crossing-free cycles and since the two edges $(c_1,c_4)$ and $(x,u)$ are already part of a \facialKFive, $H$ contains only the edges of the four identified \facialKFives and possibly one or both of edges $(c_0,z)$ and $(x',y)$.

Consider a relabeling of the vertices of $H$ based on the following table:
\smallskip\\
%
%\begin{table}[h]
%\centering
\begin{minipage}{\textwidth}\centering{\begin{tabular}{l|c|c|c|c|c|c|c|c|c|c}
Figure~\ref{fig:crossingEdgeInAK5:18}\hphantom{a}               & \hphantom{a}{$c_0$}\hphantom{a} & \hphantom{a}$c_1$\hphantom{a} & \hphantom{a}$c_2$\hphantom{a} & \hphantom{a}$c_3$\hphantom{a} & \hphantom{a}$c_4$\hphantom{a} & \hphantom{a}$u$\hphantom{a}  & \hphantom{a}$v=y$\hphantom{a} & \hphantom{a}$w=z$\hphantom{a} & \hphantom{a}$x$\hphantom{a} & \hphantom{a}$x'$\hphantom{a} \\ \hline
 $G_{CB}$ \hphantom{a}& $v$   & $x'$  & $u'$  & $u$   & $x$   & $y'$ & $w'$  & $v'$  & $y$ & $w$ 
\end{tabular}}
\end{minipage}
\smallskip

%\end{table}
We obtain the \crossingConfiguration substructure of Fig.~\ref{fig:crossingConfigurationStart} where $C$ is the \facialKFive $\langle v,x',u',u, x \rangle$ as claimed. Since the boundaries of the four identified \facialKFives are crossing-free in $\emb$, we can easily verify that the edges of $H$ are classified to \crossing and potentially planar according to the definition of \crossingConfiguration; see also Fig.~\ref{fig:crossingConfigurationStart}. 
\end{proof}

\setcounter{section}{2}
\renewcommand\thesection{\Alph{section}}